\documentclass[10pt,twocolumn,twoside]{IEEEtran}
\usepackage{cite}
\usepackage{graphicx}
\usepackage{subcaption}
\usepackage[cmex10]{amsmath}
\usepackage{algpseudocode}
\usepackage{algorithmicx} 
\usepackage{array}
\usepackage{url}
\usepackage{caption}
\usepackage{epsfig}
\usepackage{amsfonts,amssymb,multirow,bigstrut,booktabs,ctable,latexsym}

\usepackage{cuted}
\usepackage{flushend}

\usepackage[mathscr]{eucal}
\usepackage{verbatim}
\usepackage{amsthm}
\usepackage{algorithm}

\begin{document}

\newtheorem{definition}{Definition}
\newtheorem{lemma}{Lemma}
\newtheorem{theorem}{Theorem}
\newtheorem{example}{Example}
\newtheorem{proposition}{Proposition}
\newtheorem{remark}{Remark}
\newtheorem{assumption}{Assumption}
\newtheorem{corollary}{Corollary}
\newtheorem{property}{Property}
\newtheorem{ex}{EX}
\newtheorem{problem}{Problem}
\newcommand{\argmin}{\arg\!\min}
\newcommand{\argmax}{\arg\!\max}
\newcommand{\st}{\text{s.t.}}
\newcommand \dd[1]  { \,\textrm d{#1}  }

\newcounter{mytempeqncnt}
\title{Controlled Islanding and Reconnection with Stability Guarantees via\\ Submodular Optimization}
\author{Shiyu Cheng, \IEEEmembership{Graduate Student Member, IEEE}, and Andrew Clark, \IEEEmembership{Senior Member, IEEE}
\thanks{S. Cheng and A. Clark are with the Department of Electrical and Systems Engineering, Washington University in St. Louis, St. Louis, MO, USA. Email: \{cheng.shiyu, andrewclark\}@wustl.edu}
}
\maketitle
%%%%%%%%%%%%%%%%%%%%%%%%%%%%%%%%%%%%%%%%%%%%%%%%%%%%%%%%%%%%%%%%%%%%%
\begin{abstract}
%%%%%%%%%%%%%%%%%%%%%%%%%%%%%%%%%%%%%%%%%%%%%%%%%%%%%%%%%%%%%%%%%%%%%
In large-scale networked dynamical systems, local faults or disturbances
may propagate through the interconnections and compromise the stability
of the entire network. Partitioning the network into disjoint
subsystems, or islands, limits this propagation by isolating unstable
portions of the network from those that remain stable.
This paper studies a two-stage controlled islanding and reconnection
problem. In the first stage, the network is partitioned into stable
and unstable islands. In the second stage, selected stable islands
are reconnected to recover connectivity lost during islanding while
satisfying stability requirements. We derive Lyapunov-based sufficient
conditions for both stages. To address the coupling between the
islanding and reconnection decisions, we construct an augmented graph
on which the bases of a graphic matroid induce exactly the admissible
islanding and reconnection pairs. We further show that the resulting
optimization problem admits an equivalent nonincreasing supermodular
formulation. Based on this structure, we develop a local-search
algorithm with counterexample refinement and establish a performance
guarantee for its inner loop with fixed sample sets. Numerical studies
on linear networked systems illustrate the proposed procedure, examine
its sensitivity to fault locations, and compare the proposed method
with a mixed-integer linear programming benchmark in terms of solution
quality and computational performance.
\end{abstract}


\section{Introduction}
\label{section:introduction}
Large-scale networked dynamical systems arise in applications such as
power networks, multi-agent systems, and robotic networks \cite{dorfler2014synchronization,olfati2007consensus}. Their operation depends
critically on interactions among subsystems, through which faults,
disturbances, or other perturbations may propagate and compromise
overall system stability \cite{busby2021cascading}. When the original interconnection
structure cannot be maintained safely, network reconfiguration
provides a mechanism for mitigating the effects of such perturbations.

Controlled islanding is an important form of network reconfiguration,
particularly in power systems \cite{xu2010controlled, amraee2017controlled}. It removes selected network edges
to partition the system into disjoint components, or islands, thereby
isolating unstable portions of the network from those that remain
stable. A central challenge is to ensure the dynamical stability of
the resulting islands. Our preliminary conference work \cite{cheng2023submodular} addressed
this problem by incorporating Lyapunov-based stability conditions
directly into the islanding decision. That work, however, terminated
once the stable and unstable islands had been formed and did not
consider subsequent reconnection.

Islanding alone may leave the network unnecessarily fragmented. After
islanding, the states of stable islands evolve toward their respective
asymptotically stable equilibria, after which some of the removed
connections may be restored without sacrificing stability. This
motivates a two-stage network-reconfiguration problem: the first stage
isolates unstable portions of the network while preserving stable
islands, and the second reconnects selected stable islands to recover
network connectivity. The two stages are coupled because the edges
available for reconnection depend on the preceding islanding decision,
while the restored couplings may alter the stability of the
reconnected components.

Building on our preliminary work, we develop a unified optimization
framework for stability-aware islanding and reconnection of nonlinear
networked dynamical systems. The main contributions are as follows:
\begin{itemize}
    \item We formulate a two-stage controlled-islanding and
    reconnection problem and derive Lyapunov-based sufficient
    conditions that enforce stability of both post-islanding islands
    and reconnected components.

    \item We construct an augmented-graph representation that encodes
    the coupled islanding and reconnection decisions as the bases of a
    graphic matroid. We further show that the resulting objective
    admits a nonincreasing supermodular reformulation, yielding a
    local-search algorithm with a provable performance guarantee.

    \item We evaluate the proposed method on networked linear systems,
    illustrating the islanding and reconnection procedure and
    comparing solution quality and computational performance with a
    mixed-integer linear programming baseline.
\end{itemize}

The remainder of this paper is organized as follows.
Section~\ref{section:related-work} reviews related work.
Section~\ref{section:preliminaries} introduces the system model and
required preliminaries.
Section~\ref{section:problem-formulation} formulates the islanding and
reconnection problem and derives stability conditions.
Section~\ref{sec:matroid} develops the matroid and supermodular
reformulation and the proposed solution algorithm.
Section~\ref{section:simulation} presents the numerical results, and
Section~\ref{section:conclusion} concludes the paper.

\section{Related Work}
\label{section:related-work}
Controlled islanding has been studied extensively in power systems
using approaches including slow coherency
\cite{you2004slow,xu2009slow}, spectral clustering
\cite{ding2012two}, mixed-integer optimization
\cite{patsakis2019strong,teymouri2019milp,esmaili2020convex},
and submodular optimization
\cite{liu2018controlled,sahabandu2022submodular,sahabandu2022hybrid}.
Mixed-integer formulations can directly incorporate connectivity,
coherency, power-flow, and load-generation constraints, but may become
computationally expensive as the network size grows. Submodular
approaches instead exploit combinatorial structure in islanding
objectives to obtain scalable algorithms with provable performance
guarantees. Beyond controlled islanding, submodular optimization has
also been applied to networked control problems such as system
partitioning, sensor placement, and epidemic intervention
\cite{kazm2026partitioning,cheng2024modeling}.

A separate challenge is to certify the dynamical stability of the
resulting islands. Many optimization-based islanding methods use
quantities correlated with stability and assess the resulting
partition through subsequent dynamical analysis or simulation \cite{you2004slow,xu2009slow}.
Energy-function and Lyapunov-based methods provide a more direct means
of stability certification, although incorporating such conditions
into the combinatorial islanding decision is challenging. Our
preliminary work \cite{cheng2023submodular} addressed this issue by constructing an
islanding objective from Lyapunov-based stability conditions and
showing that it admits a graphic-matroid and supermodular
representation. Related submodular islanding formulations have
incorporated post-disturbance operating requirements such as
load-generation balance and transmission-line capacity, but remain
focused on the islanding decision~\cite{sahabandu2022hybrid,niu2023hybrid}.

The stability of interconnected subsystems has been studied using
Lyapunov-based, small-gain, and compositional methods \cite{dashkovskiy2010small,topcu2009compositional}. These
methods characterize how subsystem couplings affect stability, but in
the present setting the couplings introduced by reconnection are
themselves decision variables. Post-islanding reconnection and restoration have also been studied,
with restoration considerations such as blackstart availability
incorporated into island formation \cite{demetriou2018real}. In contrast, we explicitly optimize the
reconnection decision after islanding while imposing dynamical
stability requirements on the resulting reconnected components.
\section{Preliminaries}
\label{section:preliminaries}
This section introduces the notation and system model used throughout the paper, including the effects of islanding and reconnection on the system dynamics. We then review the necessary background on submodularity, matroids, and graph theory.
\subsection{Notation}

We use $\mathbb{N}_+$ to denote the set of positive integers,
$|\cdot|$ to denote set cardinality, $\|\cdot\|$ to denote the
Euclidean norm, and $0$ to denote a zero vector of appropriate
dimension. For vectors $v_i\in\mathbb{R}^{n_i}$, $i=1,\ldots,n$, we
define
    $\operatorname{col}(v_1,\ldots,v_n)
    \triangleq
    [v_1^T,\ldots,v_n^T]^T$.
For an index set $I=\{i_1,\ldots,i_k\}$ with
$i_1<\cdots<i_k$, we use
   $ \operatorname{col}(\{v_i:i\in I\})
    \triangleq
    \operatorname{col}(v_{i_1},\ldots,v_{i_k})$.
For
$x=\operatorname{col}(x_1,\ldots,x_N)$, where
$x_i\in\mathbb{R}^{p_i}$, and
$I=\{i_1,\ldots,i_k\}\subseteq\{1,\ldots,N\}$ with
$i_1<\cdots<i_k$, let
$p_I\triangleq\sum_{i\in I}p_i$ and define
    $\operatorname{proj}_I(x)
    \triangleq
    \operatorname{col}(x_{i_1},\ldots,x_{i_k})
    \in\mathbb{R}^{p_I}$.
For a set
$\Omega\subseteq\prod_{i=1}^N\mathbb{R}^{p_i}$, define
    $\operatorname{proj}_I(\Omega)
    \triangleq
    \{\operatorname{proj}_I(x):x\in\Omega\}$.
For $z\in\mathbb{R}^n$ and $\epsilon>0$, let
    $\mathbb{B}(z,\epsilon)
    \triangleq
    \{y\in\mathbb{R}^n:\|y-z\|\leq\epsilon\}$.

\subsection{System, Islanding, and Reconnection Models}
We consider a network of $N$ nodes with node set $\mathcal{V} \triangleq \{1, \ldots, N\}$ and edge set $\mathcal{E}\subseteq \mathcal{V} \times \mathcal{V}$. The graph $\mathcal{G} \triangleq (\mathcal{V}, \mathcal{E})$ is connected and undirected. We define $M \triangleq |\mathcal{E}|$ as the number of edges in $\mathcal{G}$. We define $\mathcal{N}_i \triangleq \{j\in \mathcal{V}: (i,j)\in \mathcal{E}\}$ as the neighbor set of node $i$. 
For each node $i\in \mathcal{V}$, let $x_i(t)\in \mathcal{X}_i\subseteq \mathbb{R}^{p_i}$ denote its state variable, where $p_i\in \mathbb{N}_{+}$. The dynamics of node $i$ are governed by
\begin{align}
\label{eq:dynamics-pre-islanding}
    \dot{x}_i(t) &= f_i(x_i(t)) + \sum_{j\in \mathcal{N}_i}f_{ij}(x_i(t), x_j(t)),
\end{align}
where $f_i:\mathbb{R}^{p_i}\rightarrow \mathbb{R}^{p_i}$ and $f_{ij}: \mathbb{R}^{p_i}\times \mathbb{R}^{p_j}\rightarrow \mathbb{R}^{p_i}$ are locally Lipschitz.
We define islanding as follows.
\begin{definition}[Islanding]
\label{def:islanding}
A valid islanding of $\mathcal{G}$ is a collection
$\mathcal{I}\triangleq\{I_1,\ldots,I_m\}$ for some
$m\in\{1,\ldots,N\}$, where, for each
$r\in\{1,\ldots,m\}$,
$I_r\triangleq(\mathcal{V}_r,\mathcal{E}_r)$ and
$\mathcal{E}_r\triangleq
\{(i,j)\in\mathcal{E}:i,j\in\mathcal{V}_r\}$,
and the following conditions are satisfied:
\begin{enumerate}
    \item[(i)] $\mathcal{V}_r \neq \emptyset$ for all $r \in \{1, \ldots, m\}$, and $\mathcal{V}_r \cap \mathcal{V}_s = \emptyset$ for all $r, s \in \{1, \ldots, m\}$ with $r \neq s$;
    \item[(ii)] $\bigcup_{r=1}^{m} \mathcal{V}_r = \mathcal{V}$;
    \item[(iii)] $I_r$ is connected for all $r\in \{1, \ldots, m\}$.
\end{enumerate}
\end{definition}
We assume that, when islanding occurs, the state $x_i(t)$ of node $i$ does not change, but the dynamics $\dot{x}_i(t)$ in the post-islanding system instantaneously change to 
\begin{align}
\label{eq:dynamics-post-islanding}
    \dot{x}_i(t) &= f_i(x_i(t)) + \sum_{j\in \mathcal{N}_i\cap\mathcal{V}_r}f_{ij}(x_i(t), x_j(t))
\end{align}
where $r$ is the unique index with $i\in \mathcal{V}_r$.
We next present the model of how edges are reconnected after islanding has taken place. For an islanding strategy $\mathcal{I}$, we use 
\begin{align}
    \label{eq:ec}
    \mathcal{E}^{\mathrm{cut}}(\mathcal{I}) \triangleq \mathcal{E}\setminus (\cup_{r=1}^{m}\mathcal{E}_r )   
\end{align}
to denote the edges that are removed in the islanding stage. 
By Definition~\ref{def:islanding}, every edge of $\mathcal{E}^{cut}(\mathcal{I})$ has its two endpoints in distinct islands.
We select a subset $\mathcal{E}^\mathrm{rec}\subseteq \mathcal{E}^\mathrm{cut}(\mathcal{I})$ to reconnect. For each node $i\in \mathcal{V}$, we denote
\begin{align}
    \mathcal{N}_i^{\mathrm{rec}}\triangleq \{j\in \mathcal{N}_i: (i,j)\in \mathcal{E}^{\mathrm{rec}}\}
\end{align}
as the set of reconnected neighbors of node $i$. If $i\in \mathcal{V}_r$, then $\mathcal{N}_i^{\mathrm{rec}}\cap \mathcal{V}_r = \emptyset$, hence, $(\mathcal{N}_i\cap \mathcal{V}_r)\cap \mathcal{N}_i^{\mathrm{rec}} = \emptyset$.
We assume that the dynamics $\dot{x}_i(t)$ in the reconnected system instantaneously change to 
\begin{align}
\label{eq:dynamics-reconnection}
    \dot{x}_i(t) = f_i(x_i(t)) &+ \sum_{j\in \mathcal{N}_i\cap\mathcal{V}_r}f_{ij}(x_i(t), x_j(t))\nonumber \\
    &+\sum_{j\in \mathcal{N}_i^{\mathrm{rec}}}f_{ij}(x_i(t), x_j(t)),
\end{align}
where $r$ is the unique index with $i\in \mathcal{V}_r$.

\subsection{Submodular Functions and Matroids}
\label{sec:prelim-submodular}

For any finite set $X$, a function
$h:2^X\rightarrow\mathbb{R}$ is submodular if, for any
$S\subseteq T\subseteq X$ and $v\in X\setminus T$,
\[
    h(S\cup\{v\})-h(S)
    \geq
    h(T\cup\{v\})-h(T).
\]
A function $h$ is supermodular if $-h$ is submodular.
It is monotone nonincreasing (resp. nondecreasing) if
$h(S)\geq h(T)$ (resp. $h(S)\leq h(T)$) whenever
$S\subseteq T$. Nonnegative weighted sums of submodular
(resp. supermodular) functions are submodular
(resp. supermodular) and preserve monotonicity when the
constituent functions have the same monotonicity. Furthermore,
if $h$ is nonincreasing and supermodular, then
$\max\{h(S),c\}$ is nonincreasing and supermodular for any
constant $c$.
We next recall standard matroid terminology
\cite{oxley2011matroid}.

\begin{definition}[Matroid]
A tuple $\mathcal{M}=(X,\mathcal{J})$, where $X$ is a finite
ground set and $\mathcal{J}\subseteq 2^X$, is a matroid if:
\begin{enumerate}
    \item[(i)] $\emptyset\in\mathcal{J}$;
    \item[(ii)] if $A\subseteq B$ and $B\in\mathcal{J}$, then
    $A\in\mathcal{J}$;
    \item[(iii)] if $A,B\in\mathcal{J}$ and $|A|<|B|$, then
    there exists $v\in B\setminus A$ such that
    $A\cup\{v\}\in\mathcal{J}$.
\end{enumerate}
The sets in $\mathcal{J}$ are called independent sets.
A maximal independent set is called a basis, and the family
of bases is denoted by $\mathcal{B}(\mathcal{M})$.
\end{definition}

\begin{definition}[Rank function]
For a matroid $\mathcal{M}=(X,\mathcal{J})$, its rank function is
   $ \rho_{\mathcal{M}}(S)
    \triangleq
    \max\bigl\{|R|:R\subseteq S,\ R\in\mathcal{J}\bigr\}$ for all $S\subseteq X$.
A set $S\subseteq X$ is independent if and only if
$\rho_{\mathcal{M}}(S)=|S|$. Moreover,
$\rho_{\mathcal{M}}$ is monotone nondecreasing and submodular.
\end{definition}

We will use the following standard properties of matroid rank
functions and bases.

\begin{lemma}[\cite{oxley2011matroid}]
\label{lemma:unit-increase}
For any $S\subseteq X$ and $v\in X$, adding $v$ increases the
matroid rank by at most one, i.e.,
$\rho_{\mathcal{M}}(S\cup\{v\})-\rho_{\mathcal{M}}(S)\in\{0,1\}$.
\end{lemma}

\begin{lemma}[\cite{oxley2011matroid}]
\label{lemma:basis-exchange-bijection}
Let $B,B^{\prime}\in\mathcal{B}(\mathcal{M})$. There exists a
bijection
    $\pi:B\setminus B^{\prime}
    \rightarrow
    B^{\prime}\setminus B$
such that
    $B\setminus\{v\}\cup\{\pi(v)\}
    \in\mathcal{B}(\mathcal{M})$
for every $v\in B\setminus B^{\prime}$.
\end{lemma}

\subsection{Labeled Multigraphs and Graphic Contraction}
\label{subsec:multigraph_contraction}

We recall the notions of labeled multigraphs, graphic matroids, and
contraction used in Section~\ref{sec:matroid}; see,
e.g., \cite{oxley2011matroid,diestel2017graph}.

\begin{definition}[Undirected loopless labeled multigraph]
\label{def:undirected_multigraph}
An undirected loopless labeled multigraph is a tuple
    $\mathcal{G}=(\mathcal{V},\mathcal{E},\Theta)$,
where $\mathcal{V}$ is a finite vertex set, $\mathcal{E}$ is a finite
set of distinct edge elements, and
\[
    \Theta:\mathcal{E}
    \rightarrow
    \big\{\{u,v\}:u,v\in\mathcal{V},\ u\neq v\big\}
\]
assigns to each edge element its unordered pair of endpoints.
We write $\mathcal{V}(\mathcal{G})\triangleq\mathcal{V}$,
$\mathcal{E}(\mathcal{G})\triangleq\mathcal{E}$, and
$\Theta_{\mathcal{G}}\triangleq\Theta$ when the underlying multigraph
must be specified explicitly.

Distinct edge elements $e,f\in\mathcal{E}$ are called parallel if
$\Theta(e)=\Theta(f)$. Thus, parallel edges remain distinct elements
of the edge set even though they have the same endpoint pair. A simple
undirected graph is the special case in which $\Theta$ is injective.
Whenever no ambiguity arises, an edge element $e$ satisfying
$\Theta(e)=\{u,v\}$ is denoted by $(u,v)$.
\end{definition}

The distinction between an edge element and its endpoint pair will be
important because contraction may create parallel edges while
preserving the identities of the original edge elements.
For a loopless labeled multigraph
$\mathcal{G}=(\mathcal{V},\mathcal{E},\Theta)$, let
$\mathcal{J}_{\mathcal{G}}$ denote the family of acyclic subsets of
$\mathcal{E}$. Then
    $\mathcal{M}(\mathcal{G})
    \triangleq
    (\mathcal{E},\mathcal{J}_{\mathcal{G}})$
is the graphic matroid of $\mathcal{G}$ \cite{oxley2011matroid}.
If $\mathcal{G}$ is connected, the bases of
$\mathcal{M}(\mathcal{G})$ are precisely the spanning trees of
$\mathcal{G}$. In particular, two parallel edges form a two-element
circuit of $\mathcal{M}(\mathcal{G})$ and therefore cannot both belong
to an independent set.

We next recall matroid contraction. Let
$\mathcal{M}=(\mathcal{X},\mathcal{J})$ be a matroid and let
$e\in\mathcal{X}$ be a non-loop. The contraction $\mathcal{M}/e$ is
the matroid on $\mathcal{X}\setminus\{e\}$ whose independent sets are
\[
    \mathcal{J}(\mathcal{M}/e)
    =
    \big\{
        I\subseteq\mathcal{X}\setminus\{e\}:
        I\cup\{e\}\in\mathcal{J}
    \big\}.
\]
Equivalently,
    $\mathcal{B}(\mathcal{M}/e)
    =
    \big\{
        B\setminus\{e\}:
        B\in\mathcal{B}(\mathcal{M}),\ e\in B
    \big\}$.

We also define contraction for the labeled multigraphs used in this
paper. Let
$\mathcal{H}=(\mathcal{V},\mathcal{E},\Theta_{\mathcal{H}})$ be an
undirected loopless labeled multigraph, and let
$e\in\mathcal{E}$ have endpoints
$\Theta_{\mathcal{H}}(e)=\{u,v\}$. Assume that no other edge of
$\mathcal{H}$ is parallel to $e$. Let $r$ be a new vertex and define
   $ q:\mathcal{V}
    \rightarrow
    \big(\mathcal{V}\setminus\{u,v\}\big)\cup\{r\}$
by $q(u)=q(v)=r$,
    $q(x)=x$, $x\in\mathcal{V}\setminus\{u,v\}$.
    The contracted multigraph $\mathcal{H}/e$ has vertex set
$\mathcal{V}(\mathcal{H}/e)
=
(\mathcal{V}\setminus\{u,v\})\cup\{r\}$,
edge set
$\mathcal{E}(\mathcal{H}/e)
=
\mathcal{E}\setminus\{e\}$,
and endpoint map
\[
    \Theta_{\mathcal{H}/e}(f)
    =
    \{q(x):x\in\Theta_{\mathcal{H}}(f)\},
    \qquad
    f\in\mathcal{E}\setminus\{e\}.
\]
Since no edge other than $e$ has endpoint pair $\{u,v\}$,
$\Theta_{\mathcal{H}/e}(f)$ contains two distinct vertices for every
$f\in\mathcal{E}\setminus\{e\}$; hence $\mathcal{H}/e$ remains
loopless. Contraction may nevertheless create parallel edges, and all
edge elements other than $e$ retain their identities.

The following standard result relates graph contraction and matroid
contraction \cite{oxley2011matroid}.

\begin{lemma}[Graphic contraction]
\label{lem:graphic_contraction}
Let $\mathcal{H}$ be a connected undirected loopless labeled
multigraph, and let $e\in\mathcal{E}(\mathcal{H})$ be an edge that is
not parallel to any other edge. Then
\[
    \mathcal{M}(\mathcal{H}/e)
    =
    \mathcal{M}(\mathcal{H})/e.
\]
Consequently, the map
    $T\longmapsto T\setminus\{e\}$
is a bijection from the spanning trees of $\mathcal{H}$ containing
$e$ to the spanning trees of $\mathcal{H}/e$, with inverse
$S\longmapsto S\cup\{e\}$.
\end{lemma}

\section{Problem Formulation}
\label{section:problem-formulation}
Under the nominal dynamics, the networked system operates near an
asymptotically stable equilibrium. A disturbance or fault may drive
the system outside its region of attraction. We therefore partition
the network into stable and unstable islands, where stability is
determined by the post-islanding dynamics. The islanding stage seeks
to minimize the costs associated with unstable nodes and removed
edges while ensuring the stability of the stable islands.

After islanding, each stable island evolves toward its post-islanding
equilibrium. Once the states of the stable islands enter suitable
neighborhoods of these equilibria, a subset of the removed edges may
be reconnected to reduce network fragmentation. The reconnection
stage requires the resulting components to remain stable. This two-stage procedure is illustrated in Fig.~\ref{fig:islanding-reconnection}. We next
formalize this two-stage islanding and reconnection problem.
\begin{figure*}[!htbp]
    \centering

    \begin{subfigure}[t]{0.32\textwidth}
        \centering
        \includegraphics[width=\linewidth]{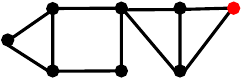}
        \caption{Original system}
        \label{fig:subfig1}
    \end{subfigure}
    \hfill
    \begin{subfigure}[t]{0.32\textwidth}
        \centering
        \includegraphics[width=\linewidth]{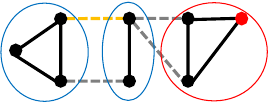}
        \caption{Post-islanding system}
        \label{fig:subfig2}
    \end{subfigure}
    \hfill
    \begin{subfigure}[t]{0.32\textwidth}
        \centering
        \includegraphics[width=\linewidth]{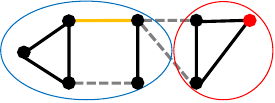}
        \caption{Reconnected System}
        \label{fig:subfig3}
    \end{subfigure}
    \caption{Illustration of the islanding and reconnection procedure.
(a) Original system, with the faulty node shown in red.
(b) Post-islanding system obtained by removing four edges, resulting
in two stable islands, indicated by blue boundaries, and one unstable
island, indicated by a red boundary. The removed edges are shown as
dashed lines, with the edge selected for reconnection shown in yellow.
(c) Reconnected system obtained by reconnecting the selected edge,
shown as a solid yellow line, thereby merging the two stable islands
into a single stable reconnected component. The remaining removed
edges are shown as dashed gray lines.}
\label{fig:islanding-reconnection}
\end{figure*}
\subsection{Problem Statement}\label{subsec:problem}
We first define the islanding stage and then introduce the
reconnection stage.

\subsubsection{Islanding} Let ${x}(t) \triangleq \operatorname{col}(x_1(t), \ldots, x_N(t))$ 
denote the state of the entire system, and let $Y 
\subset \prod_{i=1}^{N} \mathcal{X}_i$ denote the set of possible states at 
the instant of islanding. We now define stability of an island in the 
post-islanding system. 

\begin{definition}[Post-islanding stability]
\label{def:island-stability}
Let $\mathcal{I} = \{I_1, \ldots, I_m\}$ be an islanding 
of the network, and let $r \in \{1, \ldots, m\}$. Island $r$ is said to be 
\emph{stable} under $\mathcal{I}$ if, for every initial condition 
${x}(0) \in Y$, the trajectory generated by the 
post-islanding dynamics~\eqref{eq:dynamics-post-islanding} satisfies
\begin{equation}
\label{eq:island-stability}
\lim_{t \to \infty} x_i(t) = 0,~\forall i \in \mathcal{V}_r.
\end{equation}
Otherwise, island $r$ is said to be \emph{unstable}.
\end{definition}

Next, we define the cost of islanding. Let $\mathcal{S} \subseteq \{1, \ldots, m\}$ denote the indices of stable islands and 
$\mathcal{U} = \{1, \ldots, m\} \setminus \mathcal{S}$ denote the indices of unstable 
islands. The set of nodes in unstable islands is denoted as $Z \triangleq \cup_{r\in \mathcal{U}}\mathcal{V}_r$, while the set of edges in unstable islands is denoted as $ \cup_{r\in \mathcal{U}}\mathcal{E}_r$. For node $i\in Z$, we let $c_{i} \geq 0$ denote the cost of instability. For example, in a power system $c_{i}$ may represent the cost of shedding a load or shutting down a generator to prevent damage to the system. Similarly, we use $\alpha_0 c_{ij}\geq 0$ to denote the cost of removing the edge $(i,j)\in \mathcal{E}^{\mathrm{cut}}(\mathcal{I})$ with $\alpha_0>0$ and $c_{ij} \geq  0$. 
By \eqref{eq:ec}, the cost of islanding is $\sum_{i\in Z}c_i+\alpha_0 \sum_{(i,j)\in \mathcal{E}^{\mathrm{cut}}(\mathcal{I})}c_{ij}$.
\subsubsection{Reconnection}
In the reconnection stage, only removed edges whose endpoints belong to distinct stable islands of the post-islanding system are eligible for reconnection. Accordingly, define the set of eligible reconnection edges as
\begin{align}
    \label{eq:er-bar}
    {\mathcal{E}}^{\mathrm{elg}}(\mathcal{I}) \triangleq \{(i,j)\in \mathcal{E}^{\mathrm{cut}}(\mathcal{I}): i\in \mathcal{V}_{r}, &j\in \mathcal{V}_{r^{\prime}},\nonumber\\ &r, r^{\prime}\in \mathcal{S}, r\neq r^{\prime}\}.
\end{align}
The set of edges selected for reconnection satisfies $\mathcal{E}^{\mathrm{rec}}\subseteq {\mathcal{E}}^{\mathrm{elg}}(\mathcal{I})$.
By Definition~\ref{def:island-stability}, 
for every $r \in \mathcal{S}$ the states of nodes in $\mathcal{V}_r$ converge to the 
origin under the post-islanding dynamics. Consequently, there exists a neighborhood 
of the origin and a finite time after which the states of all nodes 
in $\bigcup_{r \in S} \mathcal{V}_r$ lie within that neighborhood. The 
reconnection stage is initiated at such a time, and we therefore take the 
initial condition of the reconnection system to lie in a neighborhood of the 
origin. 

Reconnecting edges in $\mathcal{E}^{\mathrm{rec}}$ alters the coupling among nodes in 
stable islands, and so the convergence guaranteed by 
Definition~\ref{def:island-stability} need not persist in the reconnection 
system. We therefore introduce a separate stability notion for this stage. 

\begin{definition}[Reconnected component]
\label{defintion:reconnected_component}
Let $\mathcal{I}$ be a valid islanding with stable-island index set
$\mathcal{S}$, and let $\mathcal{E}^{\mathrm{rec}} \subseteq \mathcal{E}^{\mathrm{elg}}(\mathcal{I})$. The
\emph{reconnected system} is the graph
\begin{equation}
  \mathcal{G}^{\mathrm{rec}} \triangleq
  \Big( \mathcal{V},\; \big(\textstyle\bigcup_{r=1}^{m}\mathcal{E}_r\big)
  \cup \mathcal{E}^{\mathrm{rec}} \Big).
\end{equation}
The \emph{reconnected components} $\mathcal{C}_1, \dots,
\mathcal{C}_{m_c}$ are the connected components of
$\mathcal{G}^{\mathrm{rec}}$ whose node sets intersect $\cup_{r \in
\mathcal{S}} \mathcal{V}_r$. We define $\mathcal{C}_s \triangleq (\mathcal{V}^c_s,
\mathcal{E}^c_s)$, and we define the index set of the islands that
$\mathcal{C}_s$ contains as
\begin{equation} \label{eq:Ss}
  \mathcal{S}_s \triangleq
  \{ r \in \mathcal{S} : \mathcal{V}_r \subseteq \mathcal{V}^c_s \},
  \quad s \in \{1,\dots,m_c\}.
\end{equation}
\end{definition}
Definition~\ref{defintion:reconnected_component} determines $\mathcal{V}^c_s$,
$\mathcal{E}^c_s$, and $\mathcal{S}_s$ from
$\mathcal{G}^{\mathrm{rec}}$ alone. The following lemma records the
properties of these sets that are used in the following sections. 
\begin{lemma} \label{lem:compstruct}
For every $s \in \{1,\dots,m_c\}$ the following hold.
\begin{enumerate}
\item[(i)] $\mathcal{S}_s \neq \emptyset$, and
$\{\mathcal{S}_1,\dots,\mathcal{S}_{m_c}\}$ is a partition of
$\mathcal{S}$. In particular $m_c \leq |\mathcal{S}|$.
\item[(ii)] $\mathcal{V}^c_s = \bigcup_{r \in \mathcal{S}_s}
\mathcal{V}_r$.
\item[(iii)] $\mathcal{E}^c_s = \big( \bigcup_{r \in \mathcal{S}_s}
\mathcal{E}_r \big) \cup \{ (i,j) \in \mathcal{E}^{\mathrm{rec}} : i,j \in \mathcal{V}^c_s \}$.
\item[(iv)] Every edge of $\mathcal{E}^{\mathrm{rec}}$ has both of its endpoints in
$\mathcal{V}^c_s$ for exactly one $s$.
\end{enumerate}
\end{lemma}
\begin{proof}
Each island ${I}_r$ is connected by
Definition~\ref{def:islanding}(iii), and $\mathcal{E}_r \subseteq
\mathcal{E}(\mathcal{G}^{\mathrm{rec}})$, so $\mathcal{V}_r$ is contained
in the node set of exactly one connected component of
$\mathcal{G}^{\mathrm{rec}}$. Since $\{\mathcal{V}_1,\dots,
\mathcal{V}_m\}$ partitions $\mathcal{V}$ by
Definition~\ref{def:islanding}(i)--(ii), the node set of every connected
component of $\mathcal{G}^{\mathrm{rec}}$ is a union of island node
sets.

Let $r \in \mathcal{U}$. Every edge of $\mathcal{E}^{\mathrm{rec}}$ lies in
$\mathcal{E}^{\mathrm{elg}}(\mathcal{I})$ and therefore has both endpoints
in stable islands, so no edge of $\mathcal{G}^{\mathrm{rec}}$ joins
$\mathcal{V}_r$ to $\mathcal{V}\setminus\mathcal{V}_r$. Hence
$\mathcal{V}_r$ is by itself the node set of a connected component, and
that component is excluded by
Definition~\ref{defintion:reconnected_component}. Consequently, each
$\mathcal{C}_s$ consists of a union of stable islands, and the islands
forming $\mathcal{V}_s^c$ are exactly those indexed by
$\mathcal{S}_s$. This proves
\[
\mathcal{V}_s^c=\bigcup_{r\in\mathcal{S}_s}\mathcal{V}_r,
\]
establishing (ii). Moreover, each $\mathcal{C}_s$ contains at least one
stable island by Definition~\ref{defintion:reconnected_component}, so
$\mathcal{S}_s\neq\emptyset$. Since distinct connected components are
node-disjoint and every stable island belongs to exactly one such
component, $\{\mathcal{S}_1,\dots,\mathcal{S}_{m_c}\}$ is a partition
of $\mathcal{S}$. Hence $m_c\leq|\mathcal{S}|$, proving (i).

For (iii), every edge of $\mathcal{E}_s^c$ belongs either to
$\bigcup_{r=1}^m\mathcal{E}_r$ or to
$\mathcal{E}^{\mathrm{rec}}$. In the former case, both endpoints belong
to the same island, which by (ii) is indexed by some
$r\in\mathcal{S}_s$. Thus
\[
\mathcal{E}_s^c
\subseteq
\left(\bigcup_{r\in\mathcal{S}_s}\mathcal{E}_r\right)
\cup
\{(i,j)\in\mathcal{E}^{\mathrm{rec}}:
i,j\in\mathcal{V}_s^c\}.
\]
Conversely, every edge in $\mathcal{E}_r$ with
$r\in\mathcal{S}_s$ has both endpoints in $\mathcal{V}_s^c$ by (ii),
and every edge of $\mathcal{E}^{\mathrm{rec}}$ whose endpoints lie in
$\mathcal{V}_s^c$ is an edge of the same connected component.
Therefore, the reverse inclusion also holds, proving (iii).

Finally, for (iv), the two endpoints of any edge in
$\mathcal{E}^{\mathrm{rec}}$ are joined by that edge in
$\mathcal{G}^{\mathrm{rec}}$ and therefore belong to the same connected
component. Since both endpoints lie in stable islands, this component
is one of $\mathcal{C}_1,\dots,\mathcal{C}_{m_c}$. Its uniqueness
follows from the fact that distinct connected components are
node-disjoint.
\end{proof}

Next, we define the stability of a reconnected component in the reconnection system. 
\begin{definition}[Reconnection stability]
\label{def:reconnection-stability}
Let $\mathcal{I}$ be an islanding with stable-island index set $\mathcal{S}$, and let 
$\mathcal{E}^{\mathrm{rec}} \subseteq {\mathcal{E}}^{\mathrm{elg}}(\mathcal{I})$ be the set of reconnected edges. A reconnected component $\mathcal{C}_s$ 
is said to be \emph{stable in the reconnection system} if, there exists $\epsilon>0$ such that under every initial 
condition $\|x_i(0)\|\leq \epsilon$ for all $i\in \mathcal{V}_s^{c}$, the trajectory 
generated by the reconnection dynamics \eqref{eq:dynamics-reconnection} satisfies
\begin{equation}
\lim_{t \to \infty} x_i(t) = 0 \quad \text{for all } i \in \mathcal{V}_s^{c}.
\end{equation}
\end{definition}

Note that if a stable island $r \in \mathcal{S}$ is not incident to any edge in 
$\mathcal{E}^{\mathrm{rec}}$, then $\mathcal{S}_s = \{r\}$ for the component $\mathcal{C}_s$ containing it, its dynamics in the reconnection system coincide with 
its dynamics in the post-islanding system, and stability is inherited from 
Definition~\ref{def:island-stability}. 

We use $\alpha_1\sum_{(i,j)\in \mathcal{E}^{\mathrm{rec}}}c_{ij}$ to denote the reward of reconnecting the edge set $\mathcal{E}^{\mathrm{rec}}$, where $0<\alpha_1<\alpha_0$.
We define
\begin{align*}
    F_1(\mathcal{I}, \mathcal{E}^{\mathrm{rec}})\triangleq  \sum_{i\in Z}c_i&+\alpha_0 \sum_{(i,j)\in \mathcal{E}^{\mathrm{cut}}(\mathcal{I})}c_{ij}- \alpha_1\sum_{(i,j)\in \mathcal{E}^{\mathrm{rec}}}c_{ij}.
\end{align*}

\begin{problem}
\label{problem:statement}
Given a graph $\mathcal{G}=(\mathcal{V},\mathcal{E})$ and node dynamics \eqref{eq:dynamics-post-islanding} and \eqref{eq:dynamics-reconnection}, compute a valid islanding $\mathcal{I}=\{I_{1},\ldots,I_{m}\}$ and a reconnected set $\mathcal{E}^{\mathrm{rec}}\subseteq \mathcal{E}^{\mathrm{elg}}(\mathcal I)$ that minimize $F_1(\mathcal{I}, \mathcal{E}^{\mathrm{rec}})$ subject to the requirements that every island $r\in \mathcal{S}$ be stable in the sense of Definition~\ref{def:island-stability} and that every reconnected component
$\mathcal{C}_s$ be stable in the sense of Definition~\ref{def:reconnection-stability}.
\end{problem}

\subsection{Optimization Formulation}\label{subsec:opt_formulation}
We use Lyapunov stability analysis to derive sufficient conditions
that guarantee the stability requirements in the islanding and
reconnection stages. To this end, we construct candidate Lyapunov
functions from node-wise and pairwise terms~\cite{dorfler2014synchronization,
chopra2009exponential,olfati2007consensus}. Accordingly, for an island
$\mathcal{I}_r$, let
$\overline{x}_r \triangleq
\operatorname{col}(\{x_i:i\in\mathcal{V}_r\})$
denote its stacked state vector. For each $i\in\mathcal{V}$, let
$W_i:\mathbb{R}^{p_i}\rightarrow\mathbb{R}$ be radially unbounded
and continuous positive definite, and, for each
$(i,j)\in\mathcal{E}$, let
$W_{ij}:\mathbb{R}^{p_i}\times\mathbb{R}^{p_j}
\rightarrow\mathbb{R}$
be continuous positive semidefinite. Define
\begin{align}
W(x)
&\triangleq
\sum_{i\in\mathcal{V}}W_i(x_i)
+
\sum_{(i,j)\in\mathcal{E}}W_{ij}(x_i,x_j),
\label{eq:W}\\
W^{r}(\overline{x}_r)
&\triangleq
\sum_{i\in\mathcal{V}_r}W_i(x_i)
+
\sum_{(i,j)\in\mathcal{E}_r}W_{ij}(x_i,x_j).
\label{eq:Wr}
\end{align}
For $\beta>0$, we define
    $\Omega_{\beta}^{r}
    \triangleq
    \left\{
    \overline{x}_r\in\mathbb{R}^{\overline{p}_r}:
    W^{r}(\overline{x}_r)\leq\beta
    \right\}$, where 
    $\overline{p}_r
    \triangleq
    \sum_{i\in\mathcal{V}_r}p_i$.
Throughout, positive definiteness of $W_i$ includes $W_i(0)=0$,
and positive semidefiniteness of $W_{ij}$ includes
$W_{ij}(0,0)=0$. We next construct a global sublevel set $\Omega$ whose projection
onto each island contains $\Omega_{\beta}^{r}$.

\begin{lemma}
\label{lemma:projection-inclusion}
Under the assumptions stated above on $W_i$ and $W_{ij}$,
suppose that, for every $(i,j)\in\mathcal{E}$, there exist
$\theta_i^{ij},\theta_j^{ij}\geq0$ such that
\[
W_{ij}(x_i,x_j)
\leq
\theta_i^{ij}W_i(x_i)+\theta_j^{ij}W_j(x_j)
\]
for all $x_i\in\mathbb{R}^{p_i}$ and
$x_j\in\mathbb{R}^{p_j}$.
Define
$\theta_i\triangleq
\max\{\theta_i^{ij}:j\in\mathcal{N}_i\}$ and
$\eta\triangleq
\max\{\deg(i)\theta_i:i\in\mathcal{V}\}$,
where $\deg(i)=|\mathcal{N}_i|$.

For any $\beta>0$, let
$\Omega\triangleq\{x:W(x)\leq(1+\eta)\beta\}$.
Let $\emptyset\neq\mathcal{L}\subseteq\mathcal{V}$ and
$\mathcal{E}_{\mathcal{L}}\subseteq
\{(i,j)\in\mathcal{E}:i,j\in\mathcal{L}\}$.
Define
$x_{\mathcal{L}}\triangleq
\operatorname{col}(\{x_i:i\in\mathcal{L}\})\in
\mathbb{R}^{p_{\mathcal{L}}}$,
where
$p_{\mathcal{L}}\triangleq\sum_{i\in\mathcal{L}}p_i$, and
\[
W_{\mathcal{L}}(x_{\mathcal{L}})
\triangleq
\sum_{i\in\mathcal{L}}W_i(x_i)
+
\sum_{(i,j)\in\mathcal{E}_{\mathcal{L}}}
W_{ij}(x_i,x_j).
\]
Let
$\Omega_{\beta}^{\mathcal{L}}
\triangleq
\{x_{\mathcal{L}}\in\mathbb{R}^{p_{\mathcal{L}}}:
W_{\mathcal{L}}(x_{\mathcal{L}})\leq\beta\}$.
Then
$\Omega_{\beta}^{\mathcal{L}}
\subseteq\operatorname{proj}_{\mathcal{L}}(\Omega)$.
\end{lemma}
\begin{proof}
Fix any
$x_{\mathcal{L}}\in\Omega_{\beta}^{\mathcal{L}}$.
Define
$\widetilde{x}\in\prod_{i\in\mathcal{V}}\mathbb{R}^{p_i}$
by
$\operatorname{proj}_{\mathcal{L}}(\widetilde{x})
=x_{\mathcal{L}}$
and $\widetilde{x}_i=0$ for all
$i\in\mathcal{V}\setminus\mathcal{L}$.
It suffices to show that $\widetilde{x}\in\Omega$.
Since $W_i(0)=0$, we have
\[
W(\widetilde{x})
=
W_{\mathcal{L}}(x_{\mathcal{L}})
+
\sum_{(i,j)\in
\mathcal{E}\setminus\mathcal{E}_{\mathcal{L}}}
W_{ij}(\widetilde{x}_i,\widetilde{x}_j).
\]
Using $W_i(0)=0$ and the assumed bound on $W_{ij}$,
\begin{align*}
&\sum_{(i,j)\in
\mathcal{E}\setminus\mathcal{E}_{\mathcal{L}}}
W_{ij}(\widetilde{x}_i,\widetilde{x}_j)\leq
\sum_{i\in\mathcal{L}}
\sum_{j\in\mathcal{N}_i}
\theta_i^{ij}W_i(x_i)\\
&\quad\leq
\sum_{i\in\mathcal{L}}
\deg(i)\theta_i W_i(x_i)
\leq
\eta\sum_{i\in\mathcal{L}}W_i(x_i).
\end{align*}
Therefore, since $W_{ij}\geq0$,
\begin{align*}
W(\widetilde{x})
&\leq
W_{\mathcal{L}}(x_{\mathcal{L}})
+\eta\sum_{i\in\mathcal{L}}W_i(x_i)\\
&\leq
(1+\eta)W_{\mathcal{L}}(x_{\mathcal{L}})
\leq
(1+\eta)\beta.
\end{align*}
Thus $\widetilde{x}\in\Omega$. Since
$\operatorname{proj}_{\mathcal{L}}(\widetilde{x})
=x_{\mathcal{L}}$, it follows that
$x_{\mathcal{L}}\in
\operatorname{proj}_{\mathcal{L}}(\Omega)$.
As $x_{\mathcal{L}}$ was arbitrary,
$\Omega_{\beta}^{\mathcal{L}}
\subseteq\operatorname{proj}_{\mathcal{L}}(\Omega)$.
\end{proof}
\begin{remark}
The edge set $\mathcal{E}_{\mathcal{L}}$ need not contain all edges
induced by $\mathcal{L}$. This flexibility is required when
Lemma~\ref{lemma:projection-inclusion} is applied to a reconnected
component. By Lemma~\ref{lem:compstruct}(iii), an edge in
$\mathcal{E}^{\mathrm{cut}}(\mathcal{I})\setminus
\mathcal{E}^{\mathrm{rec}}$ may have both endpoints in
$\mathcal{V}_s^c$ without belonging to $\mathcal{E}_s^c$.
\end{remark}
\begin{corollary}
\label{cor:proj}
Under the hypotheses and notation of
Lemma~\ref{lemma:projection-inclusion},
$\Omega_\beta^r \subseteq
\operatorname{proj}_{\mathcal V_r}(\Omega)$ for every island $r$.
\end{corollary}

\begin{proof}
Take $\mathcal L=\mathcal V_r$ and
$\mathcal E_{\mathcal L}=\mathcal E_r$ in
Lemma~\ref{lemma:projection-inclusion}. Then
$W_{\mathcal L}=W^r$, and the result follows.
\end{proof}

Next, we construct a sufficient condition ensuring the stability of the post-islanding system.
\begin{proposition}
\label{prop:island-stability}
Under the hypotheses and notation of
Lemma~\ref{lemma:projection-inclusion}, suppose, in addition, that
$W_i$ and $W_{ij}$ are continuously differentiable. If there exist
$k_1>0$ and $\beta>0$ such that
$Y\subseteq\{x:W(x)\leq\beta\}$ and
\begin{align}
\label{eq:stable-Lyapunov}
\sum_{i\in\mathcal{V}_r}
\frac{\partial W^r}{\partial x_i}
\left(
f_i(x_i)
+
\sum_{j\in\mathcal{N}_i\cap\mathcal{V}_r}
f_{ij}(x_i,x_j)
\right)
\leq -k_1 W^r(\overline{x}_r)
\end{align}
for all $x\in\Omega$, where
$\Omega\triangleq\{x:W(x)\leq(1+\eta)\beta\}$,
then island $\mathcal{I}_r$ is stable.
\end{proposition}

\begin{proof}
Fix an arbitrary $x(0)\in Y$. Since
$Y\subseteq\{x:W(x)\leq\beta\}$ and all terms defining $W$
are nonnegative, we have 
   $ W^r(\overline{x}_r(0))
    \leq W(x(0))
    \leq \beta$.
Hence $\overline{x}_r(0)\in\Omega_\beta^r$.

We next show that the Lyapunov inequality holds throughout
$\Omega_\beta^r$. Fix any
$\overline{x}_r\in\Omega_\beta^r$. By
Corollary~\ref{cor:proj}, there exists $\widetilde{x}\in\Omega$
such that
$\operatorname{proj}_{\mathcal V_r}(\widetilde{x})
=\overline{x}_r$.
Evaluating \eqref{eq:stable-Lyapunov} at $\widetilde{x}$ and
noting that both sides of \eqref{eq:stable-Lyapunov} depend only
on the states in $\mathcal V_r$, we obtain, along the
post-islanding dynamics,
   $ \dot{W}^r(\overline{x}_r)
    \leq -k_1 W^r(\overline{x}_r)$,
    $\forall \overline{x}_r\in\Omega_\beta^r$.
It follows that $\Omega_\beta^r$ is forward invariant.
Moreover, the comparison lemma gives
   $W^r(\overline{x}_r(t))
    \leq
    W^r(\overline{x}_r(0))e^{-k_1t}$, for all 
   $t\geq0$.
Thus $W^r(\overline{x}_r(t))\to0$ as $t\to\infty$.
By construction,
$W^r(\overline{x}_r)\geq
\sum_{i\in\mathcal V_r}W_i(x_i)$; hence $W^r$ is positive
definite and radially unbounded. Therefore
$\overline{x}_r(t)\to0$ as $t\to\infty$, and consequently
$x_i(t)\to0$ for every $i\in\mathcal V_r$.
Since $x(0)\in Y$ was arbitrary, $\mathcal I_r$ is stable.
\end{proof}

Next, we analyze the stability of the reconnection system. For a
reconnected component
$\mathcal{C}_s=(\mathcal{V}_s^c,\mathcal{E}_s^c)$, define
$\widehat{x}_s\triangleq
\operatorname{col}(\{x_i:i\in\mathcal{V}_s^c\})$ and
\begin{align}
V^s(\widehat{x}_s)
\triangleq
\sum_{i\in\mathcal{V}_s^c}W_i(x_i)
+
\sum_{(i,j)\in\mathcal{E}_s^c}W_{ij}(x_i,x_j).
\label{eq:Vr}
\end{align}
By Lemma~\ref{lem:compstruct}, $V^s$ admits the decomposition
\begin{align}
V^s(\widehat{x}_s)
=
\sum_{r\in\mathcal{S}_s}
\left(
W^r(\overline{x}_r)
+
\frac{1}{2}
\sum_{i\in\mathcal{V}_r}
\sum_{j\in\mathcal{N}_i^{\mathrm{rec}}}
W_{ij}(x_i,x_j)
\right).
\label{eq:decomp}
\end{align}
Indeed, $W^r$ accounts for all node and internal-edge terms of
island $r$, while the factor $1/2$ ensures that each reconnected
edge is counted exactly once.

For any $\gamma>0$, define
$\widehat{\Omega}\triangleq
\{x:W(x)\leq(1+\eta)\gamma\}$ and
$\widehat{\Omega}_\gamma^s\triangleq
\{\widehat{x}_s:V^s(\widehat{x}_s)\leq\gamma\}$.
Applying Lemma~\ref{lemma:projection-inclusion} with
$\mathcal{L}=\mathcal{V}_s^c$,
$\mathcal{E}_{\mathcal{L}}=\mathcal{E}_s^c$, and $\beta=\gamma$
gives
$\widehat{\Omega}_\gamma^s
\subseteq
\operatorname{proj}_{\mathcal{V}_s^c}(\widehat{\Omega})$.

\begin{lemma}
\label{lem:handoff}
For any $\gamma>0$, there exists $\epsilon>0$ such that, for every
valid islanding $\mathcal{I}$, every
$\mathcal{E}^{\mathrm{rec}}
\subseteq\mathcal{E}^{\mathrm{elg}}(\mathcal{I})$, and every
resulting reconnected component
$\mathcal{C}_s=(\mathcal{V}_s^c,\mathcal{E}_s^c)$,
\[
    \|x_i\|\leq\epsilon,\quad
    \forall i\in\mathcal{V}_s^c
    \quad\Longrightarrow\quad
    \widehat{x}_s\in\widehat{\Omega}_\gamma^s .
\]
\end{lemma}

\begin{proof}
Since every reconnected component satisfies
$|\mathcal{V}_s^c|\leq N$ and $|\mathcal{E}_s^c|\leq M$,
the function $V^s$ contains at most $N+M$ nonnegative node and
edge terms. It therefore suffices to choose a neighborhood of the
origin in which each such term is bounded above by
$\gamma/(N+M)$.

For each $i\in\mathcal{V}$, continuity of $W_i$ and $W_i(0)=0$
imply the existence of $\epsilon_i>0$ such that
\[
    \|x_i\|\leq\epsilon_i
    \quad\Longrightarrow\quad
    W_i(x_i)\leq\frac{\gamma}{N+M}.
\]
Similarly, for each $(i,j)\in\mathcal{E}$, continuity of $W_{ij}$
and $W_{ij}(0,0)=0$ imply the existence of
$\epsilon_{ij}>0$ such that
\[
    \max\{\|x_i\|,\|x_j\|\}\leq\epsilon_{ij}
    \quad\Longrightarrow\quad
    W_{ij}(x_i,x_j)\leq\frac{\gamma}{N+M}.
\]
Define
   $ \epsilon\triangleq
    \min\left(
    \{\epsilon_i:i\in\mathcal{V}\}
    \cup
    \{\epsilon_{ij}:(i,j)\in\mathcal{E}\}
    \right)>0$,
where positivity follows from the finiteness of
$\mathcal{V}$ and $\mathcal{E}$.

Now consider any admissible $\mathcal{I}$,
$\mathcal{E}^{\mathrm{rec}}$, and resulting component
$\mathcal{C}_s$. If $\|x_i\|\leq\epsilon$ for every
$i\in\mathcal{V}_s^c$, then every node and edge term in $V^s$
satisfies the corresponding bound above. Hence, we have 
$V^s(\widehat{x}_s)
\leq
\frac{|\mathcal{V}_s^c|+|\mathcal{E}_s^c|}{N+M}\gamma\leq\gamma$.
It follows that
$\widehat{x}_s\in\widehat{\Omega}_\gamma^s$.
\end{proof}

Suppose every island $r\in\mathcal{S}$ is stable. Then, for each
post-islanding trajectory with $x(0)\in Y$, there exists
$T\geq0$ such that
$\|x_i(t)\|\leq\epsilon$ for every
$i\in\bigcup_{r\in\mathcal{S}}\mathcal{V}_r$ and all $t\geq T$.
Since
$\mathcal{V}_s^c=
\bigcup_{r\in\mathcal{S}_s}\mathcal{V}_r$
with $\mathcal{S}_s\subseteq\mathcal{S}$,
Lemma~\ref{lem:handoff} implies
$\widehat{x}_s(t)\in\widehat{\Omega}_\gamma^s$
for every reconnected component $\mathcal{C}_s$ and all $t\geq T$.

We now give an island-wise sufficient condition for reconnection
stability. The condition requires the additional terms induced by
the reconnected edges to be dominated by the decay margin available
from \eqref{eq:stable-Lyapunov}.

\begin{proposition}
\label{prop:reconnection}
Under the hypotheses and notation of
Lemma~\ref{lemma:projection-inclusion}, suppose, in addition, that
$W_i$ and $W_{ij}$ are continuously differentiable.
Let $\mathcal{I}$ be a valid islanding with stable-island index set
$\mathcal{S}$, let
$\mathcal{E}^{\mathrm{rec}}
\subseteq\mathcal{E}^{\mathrm{elg}}(\mathcal{I})$, and let
$\mathcal{C}_s$ be a reconnected component.
Suppose there exist constants $k_1\geq k_2>0$ and scalars
$\beta\geq\gamma>0$ such that, with
$\Omega\triangleq\{x:W(x)\leq(1+\eta)\beta\}$,
and 
$\widehat{\Omega}\triangleq
\{x:W(x)\leq(1+\eta)\gamma\}$,
the following conditions hold:
\begin{enumerate}
\item[(i)]
For every $r\in\mathcal{S}_s$, condition
\eqref{eq:stable-Lyapunov} holds for all $x\in\Omega$.

\item[(ii)]
For every $r\in\mathcal{S}_s$ and all
$x\in\widehat{\Omega}$,
\begin{align}
\label{eq:prop2}
&\sum_{i\in\mathcal{V}_r}
\frac{\partial W^r}{\partial x_i}
\sum_{k\in\mathcal{N}_i^{\mathrm{rec}}}
f_{ik}(x_i,x_k)
\nonumber\\
&+
\sum_{i\in\mathcal{V}_r}
\sum_{j\in\mathcal{N}_i^{\mathrm{rec}}}
\frac{\partial W_{ij}}{\partial x_i}
\bigg(
f_i(x_i)\nonumber\\
&\quad\quad\quad\quad\quad\quad\quad+
\sum_{k\in
(\mathcal{N}_i\cap\mathcal{V}_r)
\cup\mathcal{N}_i^{\mathrm{rec}}}
f_{ik}(x_i,x_k)
\bigg)
\nonumber\\
&\leq
-\frac{k_2}{2}
\sum_{i\in\mathcal{V}_r}
\sum_{j\in\mathcal{N}_i^{\mathrm{rec}}}
W_{ij}(x_i,x_j)
+
(k_1-k_2)W^r(\overline{x}_r).
\end{align}
\end{enumerate}
Then the reconnected component $\mathcal{C}_s$ is stable.
\end{proposition}
\begin{proof}
Since $\gamma\leq\beta$, we have
$\widehat{\Omega}\subseteq\Omega$.
For $r\in\mathcal S_s$ and $i\in\mathcal V_r$, define
$g_i^r(x)\triangleq
f_i(x_i)+
\sum_{j\in\mathcal N_i\cap\mathcal V_r}f_{ij}(x_i,x_j)$,
and 
$h_i(x)\triangleq
\sum_{j\in\mathcal N_i^{\mathrm{rec}}}f_{ij}(x_i,x_j)$,
so that the reconnection dynamics satisfy
$\dot x_i=g_i^r(x)+h_i(x)$.
Differentiating $V^s$ along the reconnection dynamics and grouping
the terms by their incident islands gives
\begin{align*}
\dot V^s
=\sum_{r\in\mathcal S_s}\Bigg[
&\sum_{i\in\mathcal V_r}
\frac{\partial W^r}{\partial x_i}g_i^r
+
\sum_{i\in\mathcal V_r}
\frac{\partial W^r}{\partial x_i}h_i\\
&\qquad\qquad+
\sum_{i\in\mathcal V_r}
\sum_{j\in\mathcal N_i^{\mathrm{rec}}}
\frac{\partial W_{ij}}{\partial x_i}
\bigl(g_i^r+h_i\bigr)
\Bigg].
\end{align*}
Here, the last two terms are precisely the contributions induced
by the reconnected edges. Since
$x\in\widehat{\Omega}\subseteq\Omega$, conditions~(i)--(ii) imply
\begin{align*}
\dot V^s
&\leq
\sum_{r\in\mathcal S_s}
\Bigg[
-k_2W^r(\overline{x}_r)
-\frac{k_2}{2}
\sum_{i\in\mathcal V_r}
\sum_{j\in\mathcal N_i^{\mathrm{rec}}}
W_{ij}(x_i,x_j)
\Bigg]\\
&=-k_2V^s(\widehat{x}_s),
\end{align*}
where the last equality follows from \eqref{eq:decomp}.

Applying Lemma~\ref{lemma:projection-inclusion} with
$\mathcal L=\mathcal V_s^c$,
$\mathcal E_{\mathcal L}=\mathcal E_s^c$, and $\beta=\gamma$
shows that every
$\widehat{x}_s\in\widehat{\Omega}_\gamma^s$
is the projection of some $x\in\widehat{\Omega}$.
Since $\mathcal C_s$ is a connected component of the reconnection
system, its dynamics, and hence $\dot V^s$, depend only on
$\widehat{x}_s$. Therefore,
\[
\dot V^s(\widehat{x}_s)
\leq-k_2V^s(\widehat{x}_s),
\qquad
\widehat{x}_s\in\widehat{\Omega}_\gamma^s.
\]

By Lemma~\ref{lem:handoff}, there exists $\epsilon>0$ such that
$\|x_i(0)\|\leq\epsilon$ for all $i\in\mathcal V_s^c$ implies
$\widehat{x}_s(0)\in\widehat{\Omega}_\gamma^s$.
On the boundary of this set,
$\dot V^s\leq-k_2\gamma<0$; hence
$\widehat{\Omega}_\gamma^s$ is forward invariant, and
\[
V^s(\widehat{x}_s(t))
\leq
V^s(\widehat{x}_s(0))e^{-k_2t}.
\]
Thus $V^s(\widehat{x}_s(t))\to0$.
Since
$V^s(\widehat{x}_s)\geq
\sum_{i\in\mathcal V_s^c}W_i(x_i)$,
$V^s$ is positive definite and radially unbounded.
Consequently, $\widehat{x}_s(t)\to0$, and
$\mathcal C_s$ is stable.
\end{proof}

The stability requirements in Problem~\ref{problem:statement} are
trajectory-based and are not directly amenable to the combinatorial
optimization formulation. We therefore impose the sufficient
Lyapunov conditions of Propositions~\ref{prop:island-stability}
and~\ref{prop:reconnection} through violation functions.
For each $r\in\mathcal S$, define
\begin{align}
\Phi_r(x;\mathcal I)
\triangleq\;&
\sum_{i\in\mathcal V_r}
\frac{\partial W^r}{\partial x_i}
\left(
f_i(x_i)
+
\sum_{j\in\mathcal N_i\cap\mathcal V_r}
f_{ij}(x_i,x_j)
\right)
\nonumber\\
&+k_1W^r(\overline{x}_r),
\label{eq:Phir}
\end{align}
and
\begin{align}
&\Psi_r(x;\mathcal I,\mathcal E^{\mathrm{rec}})
\notag\\
\triangleq\;&
\sum_{i\in\mathcal V_r}
\frac{\partial W^r}{\partial x_i}
\sum_{k\in\mathcal N_i^{\mathrm{rec}}}
f_{ik}(x_i,x_k)
\notag\\
&+
\sum_{i\in\mathcal V_r}
\sum_{j\in\mathcal N_i^{\mathrm{rec}}}
\frac{\partial W_{ij}}{\partial x_i}
\left(
f_i(x_i)
+
\sum_{k\in
(\mathcal N_i\cap\mathcal V_r)
\cup\mathcal N_i^{\mathrm{rec}}}
f_{ik}(x_i,x_k)
\right)
\notag\\
&+
\frac{k_2}{2}
\sum_{i\in\mathcal V_r}
\sum_{j\in\mathcal N_i^{\mathrm{rec}}}
W_{ij}(x_i,x_j)
-(k_1-k_2)W^r(\overline{x}_r).
\label{eq:Psir}
\end{align}
Hence, $\Phi_r\leq0$ and $\Psi_r\leq0$ are the pointwise
inequalities required by
Propositions~\ref{prop:island-stability}
and~\ref{prop:reconnection}, respectively.

We then formulate
\begin{subequations}
\label{eq:opt}
\begin{align}
\min_{\mathcal I,\mathcal E^{\mathrm{rec}}}\quad
&F_1(\mathcal I,\mathcal E^{\mathrm{rec}})
\label{eq:opt-obj}\\
\mathrm{s.t.}\quad
&F_2(\mathcal I)=0,
\label{eq:opt-island}\\
&F_3(\mathcal I,\mathcal E^{\mathrm{rec}})=0,
\label{eq:opt-recon}\\
&\mathcal E^{\mathrm{rec}}
\subseteq\mathcal E^{\mathrm{elg}}(\mathcal I),
\label{eq:opt-elg}\\
&\mathcal I\text{ is a valid islanding},
\label{eq:opt-valid}
\end{align}
\end{subequations}
where
\begin{align}
F_2(\mathcal I)
&\triangleq
\int_{\Omega}
\sum_{r\in\mathcal S}
\{\Phi_r(x;\mathcal I)\}_+\,dx,\label{eq:F2}\\
F_3(\mathcal I,\mathcal E^{\mathrm{rec}})
&\triangleq
\int_{\widehat{\Omega}}
\sum_{r\in\mathcal S}
\{\Psi_r(x;\mathcal I,\mathcal E^{\mathrm{rec}})\}_+\,dx.\label{eq:F3}
\end{align}
The following result establishes that the zero-violation constraints
in \eqref{eq:opt} enforce the corresponding pointwise Lyapunov
conditions.

\begin{lemma}
\label{lem:constraints}
Assume that $\Omega$ and $\widehat{\Omega}$ are the closures of
their interiors. Let $\mathcal I$ be a valid islanding and let
$\mathcal E^{\mathrm{rec}}
\subseteq\mathcal E^{\mathrm{elg}}(\mathcal I)$.
If \eqref{eq:opt-island} and \eqref{eq:opt-recon} hold, then every
island $r\in\mathcal S$ is stable and every reconnected component
$\mathcal C_s$ is stable.
\end{lemma}
\begin{proof}
The integrands in \eqref{eq:F2} and \eqref{eq:F3} are continuous
and nonnegative. By \eqref{eq:opt-island} and the domain regularity
assumed in the lemma, the integrand in \eqref{eq:F2} vanishes
identically on $\Omega$. Hence
$\Phi_r(x;\mathcal I)\leq0$,
$\forall x\in\Omega$ and $r\in\mathcal S$.
By the definition of $\Phi_r$ in \eqref{eq:Phir}, this inequality
is equivalent to \eqref{eq:stable-Lyapunov}. Therefore,
Proposition~\ref{prop:island-stability} implies that every
$r\in\mathcal S$ is stable.

Similarly, \eqref{eq:opt-recon} and the same domain-regularity
argument imply
$\Psi_r(x;\mathcal I,\mathcal E^{\mathrm{rec}})\leq0$,
$\forall x\in\widehat{\Omega}$, and $r\in\mathcal S$.
By the definition of $\Psi_r$ in \eqref{eq:Psir}, this inequality
is equivalent to condition~(ii) of
Proposition~\ref{prop:reconnection}. For any reconnected component
$\mathcal C_s$, Lemma~\ref{lem:compstruct} gives
$\mathcal S_s\subseteq\mathcal S$, and the preceding inequality
for $\Phi_r$ establishes condition~(i) of
Proposition~\ref{prop:reconnection} for every
$r\in\mathcal S_s$. Hence $\mathcal C_s$ is stable by
Proposition~\ref{prop:reconnection}. Since $\mathcal C_s$ is
arbitrary, every reconnected component is stable.
\end{proof}

We next penalize violations of the stability conditions and consider
\begin{subequations}
\label{eq:relaxed}
\begin{align}
\min_{\mathcal I,\mathcal E^{\mathrm{rec}}}\quad
&F_1(\mathcal I,\mathcal E^{\mathrm{rec}})
+\alpha_2F_2(\mathcal I)
+\alpha_3F_3(\mathcal I,\mathcal E^{\mathrm{rec}})
\\
\mathrm{s.t.}\quad
&\mathcal E^{\mathrm{rec}}
\subseteq\mathcal E^{\mathrm{elg}}(\mathcal I),
\\
&\mathcal I\text{ is a valid islanding},
\end{align}
\end{subequations}
where $\alpha_2,\alpha_3>0$ are penalty weights.

Problem~\eqref{eq:relaxed} remains combinatorial for two reasons.
First, the number of valid islandings grows combinatorially with the
network size. Second, the islanding and reconnection decisions are
coupled because $\mathcal E^{\mathrm{elg}}(\mathcal I)$ depends on
$\mathcal I$. These challenges motivate the solution method developed
in the next section.

\section{Matroid and Submodular Reformulation}
\label{sec:matroid}
In this section, we address the combinatorial and coupling challenges
in \eqref{eq:relaxed} by reformulating it as an optimization problem
over the bases of a graphic matroid. We construct an augmented graph
such that each spanning tree determines an admissible pair
$(\mathcal{I},\mathcal{E}^{\mathrm{rec}})$, thereby replacing the
coupled islanding and reconnection decisions by a single basis
selection. We then show that the resulting basis-constrained objective
admits a nonincreasing supermodular extension, yielding an equivalent
supermodular optimization problem and a local-search method with a
provable performance guarantee.

\subsection{Matroid Reformulation and Submodular Optimization}
\label{subsec:matroid_reform}
We first construct an augmented labeled multigraph
$\overline{\mathcal{G}}
\triangleq
(\overline{\mathcal{V}},\overline{\mathcal{E}},\overline{\Theta})$
from $\mathcal{G}=(\mathcal{V},\mathcal{E})$.
We introduce nodes $s_0$ and $d_0$, together with
$n_i^{ij}$ and $n_j^{ij}$ for every $(i,j)\in\mathcal{E}$, and define
\begin{align}
\overline{\mathcal{V}}
&\triangleq
\mathcal{V}\cup\{s_0,d_0\}
\cup
\bigl\{n_i^{ij},n_j^{ij}:(i,j)\in\mathcal{E}\bigr\},
\\
\overline{\mathcal{E}}
&\triangleq
\bigl\{(i,n_i^{ij}),(n_i^{ij},n_j^{ij}),
(n_j^{ij},j):(i,j)\in\mathcal{E}\bigr\}
\nonumber\\
&\quad\cup\{(s_0,d_0)\}
\cup
\bigl\{(s_0,i),(d_0,i):i\in\mathcal{V}\bigr\}.
\end{align}
Each pair in $\overline{\mathcal{E}}$ denotes a distinct edge element,
with endpoint map
$\overline{\Theta}((u,v))\triangleq\{u,v\}$.
The nodes $s_0$ and $d_0$ distinguish stable from unstable islands,
while $n_i^{ij}$ and $n_j^{ij}$ are used to determine the
reconnection decisions.

Let
$\mathcal{M}\triangleq\mathcal{M}(\overline{\mathcal{G}})$
be the graphic matroid of $\overline{\mathcal{G}}$, and let $e_0$
denote the edge element $(s_0,d_0)$. Define
$\mathcal{B}^{\ast}\triangleq
\{A\in\mathcal{B}(\mathcal{M}):e_0\in A\}$.
Thus, $\mathcal{B}^{\ast}$ is the family of spanning trees of
$\overline{\mathcal{G}}$ containing $e_0$.

We next construct a labeled multigraph whose basis family will be
shown to be $\mathcal{B}^{\ast}$. Using the contraction operation
introduced in Section~III-D, we first contract $e_0$ in
$\overline{\mathcal{G}}$. The resulting labeled multigraph, denoted by
$\mathcal{G}_{c}\triangleq\overline{\mathcal{G}}/e_0$,
identifies $s_0$ and $d_0$ as a single vertex $r_0$ while preserving
the identities of the remaining edge elements. Consequently, for each
$i\in\mathcal{V}$, the edge elements $(s_0,i)$ and $(d_0,i)$ become
parallel edges with common endpoint pair $\{r_0,i\}$.

A new vertex $w$ is then added to $\mathcal{G}_{c}$, and the edge
element $e_0$ is reintroduced with endpoint pair $\{r_0,w\}$, with no
other edge incident to $w$. Let $\mathcal{G}_{b}$ denote the resulting
labeled multigraph and define
$\mathcal{M}_{c}\triangleq\mathcal{M}(\mathcal{G}_{b})$.
By construction, $\overline{\mathcal{G}}$ and $\mathcal{G}_{b}$ have
the same edge-element set $\overline{\mathcal{E}}$, with the identity
of each edge element preserved.

\begin{lemma}[Graphic representation of $\mathcal{B}^{\ast}$]
\label{lem:graphic_representation}
For the graph $\mathcal{G}_{b}$ constructed above,
\[
    \mathcal{B}(\mathcal{M}_{c})=\mathcal{B}^{\ast}.
\]
Hence, $\mathcal{B}^{\ast}$ is the family of bases of a graphic
matroid on $\overline{\mathcal{E}}$.
\end{lemma}

\begin{proof}
The graph $\overline{\mathcal{G}}$ is connected by construction, and
hence so are $\mathcal{G}_{c}$ and $\mathcal{G}_{b}$. Moreover,
$e_0$ is not parallel to any other edge in either
$\overline{\mathcal{G}}$ or $\mathcal{G}_{b}$. Since $w$ is
incident only to $e_0$, every spanning tree of $\mathcal{G}_{b}$
contains $e_0$. By construction, the three labeled multigraphs satisfy
    $\mathcal{G}_{b}/e_0
    =
    \mathcal{G}_{c}
    =
    \overline{\mathcal{G}}/e_0$.
Applying Lemma~\ref{lem:graphic_contraction} to
$\mathcal{G}_{b}$ gives
\[
\mathcal{B}(\mathcal{M}_{c})
=
\bigl\{S\cup\{e_0\}:
S\in\mathcal{B}(\mathcal{M}(\mathcal{G}_{c}))\bigr\}.
\]
Applying the same lemma to $\overline{\mathcal{G}}$ gives
$\mathcal{B}^{\ast}
=
\bigl\{S\cup\{e_0\}:
S\in\mathcal{B}(\mathcal{M}(\mathcal{G}_{c}))\bigr\}$.
Therefore,
$\mathcal{B}(\mathcal{M}_{c})=\mathcal{B}^{\ast}$.
\end{proof}

The multigraph representation is necessary because, after contracting
$e_0$, the edge elements $(s_0,i)$ and $(d_0,i)$ become parallel but
remain distinct by Definition~\ref{def:undirected_multigraph}.
Consequently, by Lemma~\ref{lem:graphic_representation}, every
$A\in\mathcal{B}(\mathcal{M}_c)$ can be viewed as a spanning tree of
$\overline{\mathcal{G}}$ containing $e_0$, with the edge elements
$(s_0,i)$ and $(d_0,i)$ remaining distinguishable.
For a valid islanding $\mathcal{I}$, a pair
$(\mathcal{I},\mathcal{E}^{\mathrm{rec}})$ is called admissible if
$\mathcal{E}^{\mathrm{rec}}\subseteq
\mathcal{E}^{\mathrm{elg}}(\mathcal{I})$. Let
\[
\overline{\mathcal{R}}
\triangleq
\bigl\{(\mathcal{I},\mathcal{E}^{\mathrm{rec}}):
\mathcal{I}\text{ is a valid islanding},\
\mathcal{E}^{\mathrm{rec}}\subseteq
\mathcal{E}^{\mathrm{elg}}(\mathcal{I})\bigr\}.
\]
Thus, $\overline{\mathcal{R}}$ is the feasible set of
\eqref{eq:relaxed}.

We next characterize the islanding induced by a basis
$A\in\mathcal{B}(\mathcal{M}_c)$, viewed as the corresponding
spanning tree of $\overline{\mathcal{G}}$. Let
\[
\overline{A}\triangleq
A\cap
\bigl\{(i,n_i^{ij}),(n_i^{ij},n_j^{ij}),(n_j^{ij},j):
(i,j)\in\mathcal{E}\bigr\}.
\]
\begin{figure}[!h]
    \centering
    \includegraphics[
        width=0.6\linewidth,
        trim=0 0 0 0.2mm,
        clip
    ]{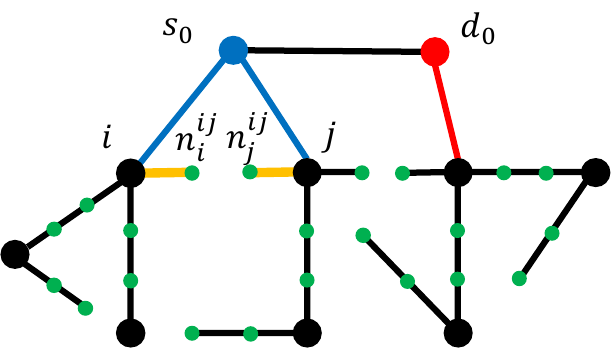}
\caption{Spanning-tree representation of the islanding and reconnection
decisions in Fig.~\ref{fig:islanding-reconnection}. A basis
$A\in\mathcal{B}(\mathcal{M}_{c})$ is shown as a spanning tree of
$\overline{\mathcal{G}}$, with edges not contained in $A$ omitted.
The green nodes are the auxiliary nodes introduced in the construction
of $\overline{\mathcal{G}}$. The blue and red edges connect the stable
and unstable islands to $s_0$ and $d_0$, respectively, and serve as
their reference edges. The two yellow edges are associated with the
original edge $(i,j)$ selected for reconnection.}
    \label{fig:basis}
\end{figure}
Two nodes $i,j\in\mathcal{V}$ belong to the same subset if and only if
they are connected by a path consisting of edges in $\overline{A}$.
This relation partitions $\mathcal{V}$ into $m(A)$ nonempty sets,
denoted by
$\mathcal{V}_1(A),\ldots,\mathcal{V}_{m(A)}(A)$. For each
$r\in\{1,\ldots,m(A)\}$, let
\begin{align*}
    \mathcal{E}_r(A)
&\triangleq
\{(i,j)\in\mathcal{E}:i,j\in\mathcal{V}_r(A)\}, \\
I_r(A)&\triangleq
(\mathcal{V}_r(A),\mathcal{E}_r(A)),
\end{align*}
and
$\mathcal{I}(A)\triangleq
\{I_1(A),\ldots,I_{m(A)}(A)\}$.
The collection $\mathcal{I}(A)$ is uniquely determined by $A$,
independently of the ordering of the island indices.

\begin{lemma}
\label{lemma:valid-islanding}
For every $A\in\mathcal{B}(\mathcal{M}_c)$,
$\mathcal{I}(A)$ is a valid islanding of $\mathcal{G}$.
Moreover, for each $r\in\{1,\ldots,m(A)\}$, exactly one edge in
$\{(s_0,i),(d_0,i):i\in\mathcal{V}_r(A)\}$ belongs to $A$.
\end{lemma}
\begin{proof}
The proof is provided in Appendix~\ref{app:proof-lemma6}.
\end{proof}
By Lemma~\ref{lemma:valid-islanding}, each island is connected to
exactly one of $s_0$ and $d_0$ by a unique edge of $A$. The endpoint
of this edge in $\mathcal{V}_r(A)$ is called the
\emph{reference node} of island $r$ and is denoted by $v_r$.
Island $r$ is labeled stable if the edge is incident to $s_0$, and
unstable otherwise.
Let $\mathcal{S}(A)$ denote the set of
indices of the islands labeled stable.

The set of reconnected edges induced by $A$ is defined as
\begin{multline}
\label{eq:er}
\mathcal{E}^{\mathrm{rec}}(A)
\triangleq
\{(i,j)\in\mathcal{E}:\exists\,r,s\in\mathcal{S}(A),\ r\neq s,\\
i\in\mathcal{V}_r(A),\
j\in\mathcal{V}_s(A),\
(n_i^{ij},n_j^{ij})\notin A\}.
\end{multline}
Thus, for an edge $(i,j)$ joining two distinct islands labeled
stable, $(i,j)\in\mathcal{E}^{\mathrm{rec}}(A)$ if and only if
$(n_i^{ij},n_j^{ij})\notin A$.
Fig.~\ref{fig:basis} illustrates the above construction for the example
in Fig.~\ref{fig:islanding-reconnection}. The connectivity induced by
$\overline{A}$ yields the three islands in
Fig.~\ref{fig:islanding-reconnection}(b). Two islands have reference
edges incident to $s_0$ and are therefore labeled stable, while the
remaining island has a reference edge incident to $d_0$ and is labeled
unstable. For the edge $(i,j)$ highlighted in yellow,
$(n_i^{ij},n_j^{ij})\notin A$. Since $i$ and $j$ belong to distinct
stable islands, \eqref{eq:er} gives
$(i,j)\in\mathcal{E}^{\mathrm{rec}}(A)$, corresponding to the
reconnection shown in Fig.~\ref{fig:islanding-reconnection}(c).
The resulting map is
\[
\Xi:\mathcal{B}(\mathcal{M}_c)\rightarrow\overline{\mathcal{R}},
\qquad
\Xi(A)\triangleq
(\mathcal{I}(A),\mathcal{E}^{\mathrm{rec}}(A)).
\]

\begin{lemma}
\label{lemma:pi-well-defined-onto}
The map $\Xi$ is well-defined and surjective. Specifically, every
$A\in\mathcal{B}(\mathcal{M}_c)$ determines a unique admissible pair
$\Xi(A)\in\overline{\mathcal{R}}$, and every
$(\mathcal{I},\mathcal{E}^{\mathrm{rec}})\in\overline{\mathcal{R}}$
is equal to $\Xi(A)$ for some
$A\in\mathcal{B}(\mathcal{M}_c)$.
\end{lemma}
\begin{proof}
For any $A\in\mathcal{B}(\mathcal{M}_c)$,
Lemma~\ref{lemma:valid-islanding} implies that
$\mathcal{I}(A)$ is a valid islanding. The same lemma guarantees a
unique reference edge for each island, and hence the stable/unstable
label of every island is uniquely determined by $A$. Consequently,
$\mathcal{E}^{\mathrm{rec}}(A)$ in \eqref{eq:er} is uniquely
determined. Moreover, every
$(i,j)\in\mathcal{E}^{\mathrm{rec}}(A)$ has endpoints in two distinct
stable islands of $\mathcal{I}(A)$. Hence
$\mathcal{E}^{\mathrm{rec}}(A)\subseteq
\mathcal{E}^{\mathrm{elg}}(\mathcal{I}(A))$, and therefore
$\Xi(A)\in\overline{\mathcal{R}}$. Thus, $\Xi$ is well-defined.

To prove surjectivity, consider any
$(\mathcal{I},\mathcal{E}^{\mathrm{rec}})
\in\overline{\mathcal{R}}$, where
$\mathcal{I}=\{I_1,\ldots,I_m\}$, and let
$\mathcal{S}$ denote its stable-island index set.
For each $r\in\{1,\ldots,m\}$, a node
$v_r\in\mathcal{V}_r$ and a spanning tree
$T_r\subseteq\mathcal{E}_r$ of $I_r$ are selected. Let
$\mathcal F\triangleq\bigcup_{r=1}^{m}T_r$.
Since $\mathcal{I}$ is a valid islanding,
$|\mathcal F|=|\mathcal{V}|-m$.
Consider the edge set
\begin{align}
\label{eq:construction}
A\triangleq\;&
\{(s_0,d_0)\}
\cup\{(s_0,v_r):r\in\mathcal{S}\}
\cup\{(d_0,v_r):r\notin\mathcal{S}\}
\nonumber\\
&\cup
\bigcup_{(i,j)\in \mathcal F}
\{(i,n_i^{ij}),(n_i^{ij},n_j^{ij}),(n_j^{ij},j)\}
\nonumber\\
&\cup
\bigcup_{(i,j)\in\mathcal{E}^{\mathrm{rec}}}
\{(i,n_i^{ij}),(n_j^{ij},j)\}
\nonumber\\
&\cup
\bigcup_{(i,j)\in
\mathcal{E}\setminus(\mathcal F\cup\mathcal{E}^{\mathrm{rec}})}
\{(n_i^{ij},n_j^{ij}),(n_j^{ij},j)\}.
\end{align}
For each $r$, the edges associated with $T_r$ connect every node in
$\mathcal{V}_r$ to $v_r$. The node $v_r$ is connected to $s_0$
when $r\in\mathcal{S}$ and to $d_0$ otherwise, while
$(s_0,d_0)\in A$. Hence all nodes in $\mathcal{V}$, together with
$s_0$ and $d_0$, are connected in $A$. For each
$(i,j)\in \mathcal F$, both auxiliary nodes $n_i^{ij}$ and $n_j^{ij}$ are
connected to $i$ and $j$. For
$(i,j)\in\mathcal{E}^{\mathrm{rec}}$, $n_i^{ij}$ is connected to
$i$ and $n_j^{ij}$ to $j$. For every remaining edge $(i,j)$, both
$n_i^{ij}$ and $n_j^{ij}$ are connected to $j$. Thus, $A$ is
connected.

Since
$\mathcal{E}^{\mathrm{rec}}\subseteq
\mathcal{E}^{\mathrm{cut}}(\mathcal{I})$
and $\mathcal F\subseteq\bigcup_{r=1}^{m}\mathcal{E}_r$, the sets
$\mathcal F$ and $\mathcal{E}^{\mathrm{rec}}$ are disjoint. Therefore,
\[
\begin{aligned}
|A|
&=m+1+3|\mathcal F|+2(|\mathcal{E}|-|\mathcal F|)\\
&=|\mathcal{V}|+2|\mathcal{E}|+1
=|\overline{\mathcal{V}}|-1.
\end{aligned}
\]
Hence $A$ is a spanning tree of $\overline{\mathcal{G}}$ containing
$e_0=(s_0,d_0)$. By
Lemma~\ref{lem:graphic_representation},
$A\in\mathcal{B}(\mathcal{M}_c)$.

It remains to verify that the constructed $A$ satisfies
$\Xi(A)=(\mathcal{I},\mathcal{E}^{\mathrm{rec}})$.
For $(i,j)\in \mathcal F$, all three edges
$(i,n_i^{ij})$, $(n_i^{ij},n_j^{ij})$, and $(n_j^{ij},j)$
belong to $A$. If $(i,j)\in\mathcal{E}^{\mathrm{rec}}$, the middle
edge $(n_i^{ij},n_j^{ij})$ is absent, whereas if
$(i,j)\in\mathcal{E}\setminus
(\mathcal F\cup\mathcal{E}^{\mathrm{rec}})$, the edge
$(i,n_i^{ij})$ is absent. Since the auxiliary nodes associated with
distinct edges of $\mathcal{G}$ are distinct, an edge
$(i,j)\in\mathcal{E}$ connects its two endpoints through
$\overline{A}$ if and only if $(i,j)\in \mathcal F$. Hence two nodes of
$\mathcal{V}$ are connected through $\overline{A}$ if and only if
they are connected through $\mathcal F$. It follows that
$\mathcal{I}(A)=\mathcal{I}$.

The reference edge of island $r$ is incident to $s_0$ exactly when
$r\in\mathcal{S}$, so the induced island labels agree with those of
$\mathcal{I}$. Moreover, an edge joining two distinct islands cannot
belong to $\mathcal F$. Hence, for every edge $(i,j)$ joining two distinct
stable islands,
\[
(n_i^{ij},n_j^{ij})\notin A
\quad\Longleftrightarrow\quad
(i,j)\in\mathcal{E}^{\mathrm{rec}}.
\]
By \eqref{eq:er},
$\mathcal{E}^{\mathrm{rec}}(A)=\mathcal{E}^{\mathrm{rec}}$.
Therefore,
$\Xi(A)=(\mathcal{I},\mathcal{E}^{\mathrm{rec}})$,
and $\Xi$ is surjective.
\end{proof}

\textit{Remark}:
The map $\Xi$ need not be injective. For a fixed admissible pair
$(\mathcal{I},\mathcal{E}^{\mathrm{rec}})$, different choices of
the reference nodes and intra-island spanning trees in the
construction above may yield distinct bases $A_1$ and $A_2$ satisfying
$\Xi(A_1)=\Xi(A_2)$. This nonuniqueness does not affect the
reformulation, since the objective on
$\mathcal{B}(\mathcal{M}_c)$ depends on $A$ only through
$\Xi(A)$.

We next express the objective terms in \eqref{eq:relaxed} as functions
of a basis $A\in\mathcal{B}(\mathcal{M}_c)$. For $i\in\mathcal{V}$,
define
\begin{align*}
\chi_i^{s}(A)
&\triangleq
\begin{cases}
1, & \text{if } (s_0,i)\in A,\\
0, & \text{otherwise},
\end{cases} \\
\chi_i^{d}(A)
&\triangleq
\begin{cases}
1, & \text{if } (d_0,i)\in A,\\
0, & \text{otherwise}.
\end{cases}
\end{align*}
For $i\in\mathcal V$ and
$j\in\overline{\mathcal V}\setminus\{s_0,d_0\}$ with $j\neq i$, define
\[
\chi_{ij}(A)
\triangleq
\begin{cases}
1, & \text{if $i$ is a reference node and $i$ and $j$ are }\\
   & \text{connected by a path using edges in $A\setminus P$},\\
0, & \text{otherwise},
\end{cases}
\]
where
$P\triangleq
\{(s_0,d_0)\}\cup
\{(s_0,k),(d_0,k):k\in\mathcal{V}\}$,
and
$\chi_{ii}(A)\triangleq\chi_i^s(A)+\chi_i^d(A)$
for all $i\in\mathcal{V}$.
By Lemma~\ref{lemma:valid-islanding}, every node
$j\in\mathcal{V}$ belongs to exactly one island with a unique
reference node, and hence $\sum_{i=1}^{N}\chi_{ij}(A)=1$ for every
$j\in\mathcal{V}$.
For $(i,j)\in\mathcal{E}$, define
\[
\phi_{ij}(A)
\triangleq
\begin{cases}
1, & (n_i^{ij},n_j^{ij})\notin A,\\
0, & \text{otherwise}.
\end{cases}
\]
For set functions $h_1,\ldots,h_q$, define
\[
\Gamma(h_1,\ldots,h_q)(A)
\triangleq
\left\{
\sum_{p=1}^{q}h_p(A)-(q-1)
\right\}_{+}.
\]
If $h_1(A),\ldots,h_q(A)\in\{0,1\}$, then
$\Gamma(h_1,\ldots,h_q)(A)
=
\prod_{p=1}^{q}h_p(A)$.
Accordingly, define
\begin{align*}
\psi_{v,i}^s(A)
&\triangleq
\Gamma(\chi_v^s,\chi_{vi})(A),\\
\psi_{v,i,j}^s(A)
&\triangleq
\Gamma(\chi_v^s,\chi_{vi},\chi_{vj})(A),\\
\psi_{v,i,j,k}^s(A)
&\triangleq
\Gamma(\chi_v^s,\chi_{vi},\chi_{vj},\chi_{vk})(A).
\end{align*}
Thus, $\psi_{v,i}^s(A)$, $\psi_{v,i,j}^s(A)$, and
$\psi_{v,i,j,k}^s(A)$ indicate, respectively, that the corresponding
one, two, or three nodes belong to the island labeled stable whose
reference node is $v$.
For $(i,j)\in\mathcal{E}$, define
\begin{equation}
\chi_{ij}^{\mathrm{rec}}(A)
\triangleq
\sum_{v=1}^{N}
\sum_{\substack{v^{\prime}=1\\v^{\prime}\neq v}}^{N}
\Gamma
\bigl(
\psi_{v,i}^s,
\psi_{v^{\prime},j}^s,
\phi_{ij}
\bigr)(A).
\label{eq:chirec}
\end{equation}
Since each node belongs to a unique island with a unique reference
node, at most one pair $(v,v^{\prime})$ in \eqref{eq:chirec} can
yield a nonzero term. Hence, for $(i,j)\in\mathcal{E}$,
$\chi_{ij}^{\mathrm{rec}}(A)=1$ if and only if $i$ and $j$ belong
to distinct islands labeled stable and
$(n_i^{ij},n_j^{ij})\notin A$. By \eqref{eq:er}, this is equivalent
to $(i,j)\in\mathcal{E}^{\mathrm{rec}}(A)$.
For $A\in\mathcal{B}(\mathcal{M}_c)$, define
$F_1(A)
\triangleq
F_1\bigl(
\mathcal{I}(A),
\mathcal{E}^{\mathrm{rec}}(A)
\bigr)$.
Using the indicators above,
\begin{align}
F_1(A)
=\;&
\sum_{i=1}^{N}\sum_{v=1}^{N}
\Gamma(\chi_v^d,\chi_{vi})(A)c_i
\nonumber\\
&+\alpha_0
\sum_{(i,j)\in\mathcal{E}}
\sum_{v=1}^{N}
\sum_{\substack{v^{\prime}=1\\v^{\prime}\neq v}}^{N}
\Gamma(\chi_{vi},\chi_{v^{\prime}j})(A)c_{ij}
\nonumber\\
&-\alpha_1
\sum_{(i,j)\in\mathcal{E}}
\chi_{ij}^{\mathrm{rec}}(A)c_{ij}.
\label{eq:F1A}
\end{align}

For each $v\in\mathcal{V}$, let
\begin{align*}
W_v(x,A)
\triangleq\;&
\sum_{i=1}^{N}\psi_{v,i}^s(A)W_i(x_i)\\
&+\frac{1}{2}\sum_{i=1}^{N}\sum_{j\in\mathcal{N}_i}
\psi_{v,i,j}^s(A)W_{ij}(x_i,x_j),
\end{align*}
and
\begin{align*}
&\Phi_v(x,A)
\triangleq\;
\sum_{i=1}^{N}
\psi_{v,i}^s(A)
\frac{\partial W_i(x_i)}{\partial x_i}f_i(x_i)\\
&+\sum_{i=1}^{N}\sum_{j\in\mathcal{N}_i}
\psi_{v,i,j}^s(A)
\bigg(
\frac{\partial W_{ij}(x_i,x_j)}{\partial x_i}f_i(x_i) \\
&
\qquad\qquad\qquad\qquad\qquad+\frac{\partial W_i(x_i)}{\partial x_i}
f_{ij}(x_i,x_j)
\bigg)\\
&+\sum_{i=1}^{N}\sum_{j\in\mathcal{N}_i}
\sum_{k\in\mathcal{N}_i}
\psi_{v,i,j,k}^s(A)
\frac{\partial W_{ij}(x_i,x_j)}{\partial x_i}
f_{ik}(x_i,x_k)\\
&+k_1W_v(x,A).
\end{align*}
If $v=v_r$ is the reference node of an island $r$ labeled stable,
then
$W_v(x,A)=W^r(x_r)$ and
$\Phi_v(x,A)=\Phi_r(x;\mathcal{I}(A))$; otherwise,
$W_v(x,A)=\Phi_v(x,A)=0$.
Therefore, for $A\in\mathcal{B}(\mathcal{M}_c)$,
\[
F_2(A)
\triangleq
F_2(\mathcal{I}(A))
=
\int_{\Omega}
\sum_{v=1}^{N}
\left\{\Phi_v(x,A)\right\}_{+}\,dx.
\]
Similarly, define
\begin{align*}
\psi^{r}_{v,i,k}(A)
&\triangleq
\Gamma
(\chi_v^s,\chi_{vi},\chi_{ik}^{\mathrm{rec}})(A),\\
\psi^{sr}_{v,i,j,k}(A)
&\triangleq
\Gamma
(\chi_v^s,\chi_{vi},\chi_{vj},
\chi_{ik}^{\mathrm{rec}})(A),\\
\psi^{rr}_{v,i,j,k}(A)
&\triangleq
\Gamma
(\chi_v^s,\chi_{vi},
\chi_{ij}^{\mathrm{rec}},
\chi_{ik}^{\mathrm{rec}})(A).
\end{align*}
For each $v\in\mathcal{V}$, let
\begin{align}
&\Psi_v(x,A)
\triangleq\;
\sum_{i=1}^{N}\sum_{k\in\mathcal{N}_i}
\psi^{r}_{v,i,k}(A)
\bigg(
\frac{\partial W_i(x_i)}{\partial x_i}
f_{ik}(x_i,x_k)\nonumber\\
&\qquad\qquad\qquad\qquad\qquad\qquad+
\frac{\partial W_{ik}(x_i,x_k)}{\partial x_i}
f_i(x_i)
\bigg)
\nonumber\\
&+
\sum_{i=1}^{N}\sum_{j\in\mathcal{N}_i}
\sum_{k\in\mathcal{N}_i}
\psi^{sr}_{v,i,j,k}(A)
\bigg(
\frac{\partial W_{ij}(x_i,x_j)}{\partial x_i}
f_{ik}(x_i,x_k)\nonumber\\
&\qquad\qquad\qquad\qquad\qquad+
\frac{\partial W_{ik}(x_i,x_k)}{\partial x_i}
f_{ij}(x_i,x_j)
\bigg)
\nonumber\\
&+
\sum_{i=1}^{N}\sum_{j\in\mathcal{N}_i}
\sum_{k\in\mathcal{N}_i}
\psi^{rr}_{v,i,j,k}(A)
\frac{\partial W_{ij}(x_i,x_j)}{\partial x_i}
f_{ik}(x_i,x_k)
\nonumber\\
&+
\frac{k_2}{2}
\sum_{i=1}^{N}\sum_{k\in\mathcal{N}_i}
\psi^{r}_{v,i,k}(A)W_{ik}(x_i,x_k)
-(k_1-k_2)W_v(x,A).
\label{eq:Psibarp}
\end{align}
If $v=v_r$ is the reference node of an island $r$ labeled stable,
then
$\Psi_v(x,A)
=
\Psi_r\bigl(
x;\mathcal{I}(A),
\mathcal{E}^{\mathrm{rec}}(A)
\bigr)$.
Otherwise, $\Psi_v(x,A)=0$.
Therefore, for
$A\in\mathcal{B}(\mathcal{M}_c)$,
\[
F_3(A)
\triangleq
F_3\bigl(
\mathcal{I}(A),
\mathcal{E}^{\mathrm{rec}}(A)
\bigr)
=
\int_{\widehat{\Omega}}
\sum_{v=1}^{N}
\left\{\Psi_v(x,A)\right\}_{+}\,dx.
\]

Let $F(A)\triangleq F_1(A)+\alpha_2F_2(A)+\alpha_3F_3(A)$ for
$A\in\mathcal{B}(\mathcal{M}_c)$.
By Lemma~\ref{lemma:pi-well-defined-onto}, the map $\Xi$ is
surjective, and the objective value associated with a basis $A$
depends only on $\Xi(A)$. Hence, \eqref{eq:relaxed} is equivalently
reformulated as
\begin{subequations}
\label{eq:matroid-formulation}
\begin{align}
\underset{A}{\min} &\quad F(A) \\
\mbox{s.t.} 
&\quad A\in \mathcal{B}(\mathcal{M}_c)
\end{align}
\end{subequations}
While \eqref{eq:matroid-formulation} has a matroid constraint, the objective function is not submodular. In what follows, we prove that \eqref{eq:matroid-formulation} is equivalent to a submodular optimization problem. We define $F(A) \triangleq F_1(A)+\alpha_2F_2(A)+\alpha_3F_3(A)$.

\begin{theorem} \label{thm:supermodular_reformulation} There exist set functions $\overline{F}_1,\overline{F}_2,\overline{F}_3: 2^{\overline{\mathcal{E}}}\rightarrow\mathbb{R}$ such that, for every $A\in\mathcal{B}(\mathcal{M}_c)$, \[ \overline{F}_1(A)=F_1(A),\; \overline{F}_2(A)=F_2(A),\;\overline{F}_3(A)=F_3(A). \] Moreover, $\overline{F}_1$, $\overline{F}_2$, and $\overline{F}_3$ are nonincreasing and supermodular on $2^{\overline{\mathcal{E}}}$. Consequently, the function \[ \overline{F}(A) \triangleq \overline{F}_1(A) +\alpha_2\overline{F}_2(A) +\alpha_3\overline{F}_3(A) \] is nonincreasing and supermodular on $2^{\overline{\mathcal{E}}}$ and satisfies  $\overline{F}(A) = F(A)$ for every $A\in\mathcal{B}(\mathcal{M}_c)$. \end{theorem}
\begin{proof} See Appendix~\ref{app:supermodular_reformulation}. 
\end{proof}
By Theorem~\ref{thm:supermodular_reformulation},
\eqref{eq:matroid-formulation} is equivalently reformulated as
\begin{subequations}
\label{eq:supermodular_problem}
\begin{align}
\underset{A}{\min} &\quad \overline{F}(A) \\
\mbox{s.t.} 
&\quad A\in \mathcal{B}(\mathcal{M}_c)
\end{align}
\end{subequations}
Although Theorem~\ref{thm:supermodular_reformulation} establishes an equivalent
nonincreasing supermodular formulation of \eqref{eq:relaxed},
evaluating $\overline{F}_2$ and $\overline{F}_3$ requires integration
over continuous state spaces and is generally impractical. We therefore
introduce a sample-based approximation.

Let
$\Omega_{\mathrm{sample}}
\subseteq \Omega$
and
$\widehat{\Omega}_{\mathrm{sample}}
\subseteq \widehat{\Omega}$
denote finite sample sets. Initially,
$\Omega_{\mathrm{sample}}$
and
$\widehat{\Omega}_{\mathrm{sample}}$
contain $P_1$ and $P_2$ points, respectively, drawn independently
and uniformly from the corresponding sets. For
$A\in\mathcal{B}(\mathcal{M}_c)$, define
\begin{equation}
\widehat{F}(A)
\triangleq
F_1(A)
+\alpha_2\widehat{F}_2(A)
+\alpha_3\widehat{F}_3(A),
\label{eq:sampled-objective}
\end{equation}
where
\begin{align*}
\widehat{F}_2(A)
&\triangleq
\frac{1}{|\Omega_{\mathrm{sample}}|}
\sum_{x\in\Omega_{\mathrm{sample}}}
\sum_{v=1}^{N}
\left\{\Phi_v(x,A)\right\}_{+},\\
\widehat{F}_3(A)
&\triangleq
\frac{1}{|\widehat{\Omega}_{\mathrm{sample}}|}
\sum_{x\in\widehat{\Omega}_{\mathrm{sample}}}
\sum_{v=1}^{N}
\left\{\Psi_v(x,A)\right\}_{+}.
\end{align*}
The constant scaling factors associated with the volumes of
$\Omega$ and $\widehat{\Omega}$ are absorbed into the penalty weights
$\alpha_2$ and $\alpha_3$.

For fixed sample sets, the pointwise constructions used in the proof
of Theorem~~\ref{thm:supermodular_reformulation} imply that
$\widehat{F}_2$ and $\widehat{F}_3$ admit nonincreasing supermodular
extensions to $2^{\overline{\mathcal{E}}}$. Consequently,
$\widehat{F}$ also admits a nonincreasing supermodular extension to
$2^{\overline{\mathcal{E}}}$. We therefore consider the sampled
optimization problem
\begin{equation}
\label{eq:sampled-matroid}
\min_{A\in\mathcal{B}(\mathcal{M}_c)}
\widehat{F}(A).
\end{equation}

For any $A\in\mathcal{B}(\mathcal{M}_c)$, define the set of feasible
single-edge exchanges by
\begin{equation*}
\mathcal{Q}(A)
\triangleq
\left\{
(e,f)\in
A\times(\overline{\mathcal{E}}\setminus A):
A\setminus\{e\}\cup\{f\}
\in\mathcal{B}(\mathcal{M}_c)
\right\}.
\end{equation*}
Moreover, define the counterexample sets
\begin{align}
\mathcal{D}_2(A)
&\triangleq
\left\{
x\in\Omega:
\sum_{v=1}^{N}
\left\{\Phi_v(x,A)\right\}_{+}>0
\right\},
\label{eq:counterexample-F2}\\
\mathcal{D}_3(A)
&\triangleq
\left\{
x\in\widehat{\Omega}:
\sum_{v=1}^{N}
\left\{\Psi_v(x,A)\right\}_{+}>0
\right\}.
\label{eq:counterexample-F3}
\end{align}
\begin{algorithm}[t]
\caption{Sample-Based Local Search for Controlled Islanding and Reconnection}
\label{algo:islanding}
\begin{algorithmic}[1]

\State Initialize $\Omega_{\mathrm{sample}}$ with $P_1$ samples
drawn independently and uniformly from $\Omega$
\State Initialize $\widehat{\Omega}_{\mathrm{sample}}$ with $P_2$
samples drawn independently and uniformly from $\widehat{\Omega}$

\While{\textbf{true}}

    \State Initialize $A\in\mathcal{B}(\mathcal{M}_c)$

    \While{\textbf{true}}

        \State Select
        $(e^{\ast},f^{\ast})
        \in
        \displaystyle
        \arg\min_{(e,f)\in\mathcal{Q}(A)}
        \widehat{F}\!\left(
        A\setminus\{e\}\cup\{f\}
        \right)$

        \State
        $A^{\prime}
        \leftarrow
        A\setminus\{e^{\ast}\}\cup\{f^{\ast}\}$

        \If{$\widehat{F}(A)>
        (1+\widehat{\epsilon})\widehat{F}(A^{\prime})$}
            \State $A\leftarrow A^{\prime}$
        \Else
            \State \textbf{break}
        \EndIf

    \EndWhile

    \If{$\mathcal{D}_2(A)=\emptyset$
    \textbf{and}
    $\mathcal{D}_3(A)=\emptyset$}
        \State \Return{$A$}
    \EndIf

    \If{$\mathcal{D}_2(A)\neq\emptyset$}
        \State Select
        $x^{\mathrm{island}}\in\mathcal{D}_2(A)$
        \State
        $\Omega_{\mathrm{sample}}
        \leftarrow
        \Omega_{\mathrm{sample}}
        \cup\{x^{\mathrm{island}}\}$
    \EndIf

    \If{$\mathcal{D}_3(A)\neq\emptyset$}
        \State Select
        $x^{\mathrm{rec}}\in\mathcal{D}_3(A)$
        \State
        $\widehat{\Omega}_{\mathrm{sample}}
        \leftarrow
        \widehat{\Omega}_{\mathrm{sample}}
        \cup\{x^{\mathrm{rec}}\}$
    \EndIf

\EndWhile

\end{algorithmic}
\end{algorithm}

Algorithm~\ref{algo:islanding} alternates between a sample-based local
search and counterexample refinement. For fixed sample sets
$\Omega_{\mathrm{sample}}$ and
$\widehat{\Omega}_{\mathrm{sample}}$, the inner loop searches over
feasible single-edge exchanges. At each iteration, one edge in the
current basis is removed and one edge outside the basis is added,
while preserving membership in $\mathcal{B}(\mathcal{M}_c)$. Among
all such exchanges, the algorithm selects one that minimizes the
sampled objective $\widehat{F}$. The current basis is updated whenever
the resulting objective value satisfies the prescribed multiplicative
improvement condition.
After the local search terminates, the algorithm checks whether the
current basis admits counterexamples associated with the stability
penalties. If $\mathcal{D}_2(A)\neq\emptyset$, a point in
$\mathcal{D}_2(A)$ is added to $\Omega_{\mathrm{sample}}$; similarly,
if $\mathcal{D}_3(A)\neq\emptyset$, a point in $\mathcal{D}_3(A)$ is
added to $\widehat{\Omega}_{\mathrm{sample}}$. The local-search
procedure is then repeated using the augmented sample sets. The
algorithm terminates when both counterexample sets are empty.
\subsection{Optimization Analysis}
\label{section:submodular-optimization}
We next analyze the performance guarantee and evaluation
complexity of the inner loop of Algorithm~\ref{algo:islanding}. The following analysis follows the local-search arguments
in~\cite{sahabandu2022hybrid,niu2023hybrid}, applied here to the sampled objective
$\widehat{F}$ and the matroid $\mathcal{M}_c$.
For fixed sample sets, let
$\widetilde{F}:2^{\overline{\mathcal{E}}}\rightarrow\mathbb{R}$
denote a nonincreasing supermodular extension of
$\widehat{F}$ established above, so that
\[
    \widetilde{F}(A)=\widehat{F}(A),
    \qquad A\in\mathcal{B}(\mathcal{M}_c).
\]
Since $\widetilde{F}$ is nonincreasing,
$\widetilde{F}(\emptyset)-\widetilde{F}(S)\geq0$ for all
$S\subseteq\overline{\mathcal{E}}$.

\begin{theorem}
\label{thm:local-search}
Fix the sample sets $\Omega_{\mathrm{sample}}$ and
$\widehat{\Omega}_{\mathrm{sample}}$, and let
    $A^\ast\in
    \arg\min_{B\in\mathcal{B}(\mathcal{M}_c)}
    \widehat{F}(B)$.
As $\widehat{\epsilon}\rightarrow0$, the inner
loop of Algorithm~\ref{algo:islanding} yields a basis
$A\in\mathcal{B}(\mathcal{M}_c)$ satisfying
\begin{equation}
    \widetilde{F}(\emptyset)-\widehat{F}(A)
    \geq
    \frac{1}{2}
    \left(
        \widetilde{F}(\emptyset)-\widehat{F}(A^\ast)
    \right).
    \label{eq:local-search-bound}
\end{equation}
\end{theorem}
\begin{proof}
For fixed sample sets, define
$\sigma(S)\triangleq \widetilde{F}(\emptyset)-\widetilde{F}(S)$.
Then $\sigma$ is nonnegative, monotone nondecreasing, and submodular.
As $\widehat{\epsilon}\to0$, termination of the inner loop implies
that $A$ is locally optimal with respect to all feasible single-edge
exchanges. Hence, by the standard $1/2$ local-search bound for
monotone submodular maximization over a matroid basis
\cite{fisher1978analysis},
$\sigma(A)\geq \frac{1}{2}
\max_{B\in\mathcal{B}(\mathcal{M}_c)} \sigma(B)
=\frac{1}{2}\sigma(A^\ast)$.
Since $\widetilde{F}(B)=\widehat{F}(B)$ for all
$B\in\mathcal{B}(\mathcal{M}_c)$, substituting the definition of
$\sigma$ yields the result.
\end{proof}
We next analyze the computational complexity of the inner
loop of Algorithm~\ref{algo:islanding}. For fixed
$\widehat{\epsilon}>0$, let
$A_0,A_1,\ldots,A_K$ denote the bases obtained after the
accepted updates, where $K$ is the total number of accepted
updates. By the update rule,
\[
    \widehat{F}(A_k)
    <
    \frac{1}{1+\widehat{\epsilon}}
    \widehat{F}(A_{k-1}),
    \qquad k=1,\ldots,K.
\]
Hence,
   $ \widehat{F}(A_K)
    <
    \frac{\widehat{F}(A_0)}
         {(1+\widehat{\epsilon})^T}$.
If $\widehat{F}(A_K)>0$, then
\begin{equation}
    K
    <
    \frac{
        \log\!\left(
            \widehat{F}(A_0)/\widehat{F}(A_K)
        \right)}
        {\log(1+\widehat{\epsilon})}.
    \label{eq:number-updates}
\end{equation}
For any $A\in\mathcal{B}(\mathcal{M}_c)$, the number of
candidate exchanges satisfies
    $|\mathcal{Q}(A)|
    \leq
    |A|\bigl(
        |\overline{\mathcal{E}}|-|A|
    \bigr)
    =
    \mathcal{O}
    \bigl(
        |\overline{\mathcal{E}}|^2
    \bigr)$.
By construction,
$|\overline{\mathcal{E}}|=3M+2N+1$, and since the original
graph is connected, $M\geq N-1$. Thus
$|\overline{\mathcal{E}}|=\mathcal{O}(M)$, and each search
evaluates at most $\mathcal{O}(M^2)$ candidate exchanges.
Since $K$ accepted updates are followed by one final search,
the total number of sampled-objective evaluations is
\begin{equation}
    \mathcal{O}\!\left(
        M^2
        \left[
            1+
            \frac{
                \log\!\left(
                    \widehat{F}(A_0)/\widehat{F}(A_K)
                \right)}
                {\log(1+\widehat{\epsilon})}
        \right]
    \right),
    \label{eq:evaluation-complexity}
\end{equation}
provided that $\widehat{F}(A_K)>0$.

This bound counts sampled-objective evaluations for one
execution of the inner loop with fixed sample sets and does
not include the counterexample-refinement iterations of the
outer loop.

\section{Numerical Studies}
\label{section:simulation}

In this section, we evaluate the proposed islanding and
reconnection framework on networked linear systems. We first
introduce the linear network model, fault setting, and the
eigenvalue-based counterexample search used in the numerical
implementation. We then present numerical results illustrating
the islanding and reconnection procedure, examining its
sensitivity to fault configurations, and comparing its solution
quality and computational performance with a MILP benchmark.

\begin{figure*}[!t]
\centering
\begin{subfigure}[t]{0.25\textwidth}
  \centering
  \includegraphics[width=\linewidth]{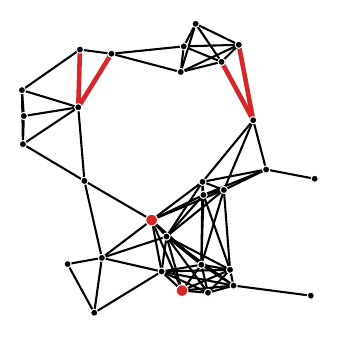}
  \caption{Original system}\label{fig:islanding_original}
\end{subfigure}
\hfill
\begin{subfigure}[t]{0.25\textwidth}
  \centering
  \includegraphics[width=\linewidth]{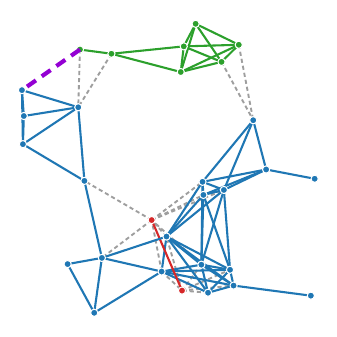}
  \caption{Post-islanding system}\label{fig:islanding_post}
\end{subfigure}
\hfill
\begin{subfigure}[t]{0.25\textwidth}
  \centering
  \includegraphics[width=\linewidth]{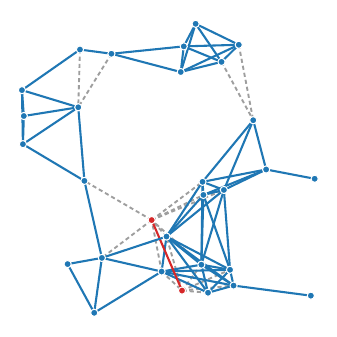}
  \caption{Reconnected system}\label{fig:islanding_reconnected}
\end{subfigure}
\caption{Islanding and reconnection on a $30$-node network with
$82$ edges. (a) Original network, with faulty nodes and edges
shown in red. (b) Post-islanding network returned by
Algorithm~\ref{algo:islanding}, consisting of two stable islands
(blue and green) and one unstable island (red). Removed edges are
shown as dashed gray lines, and the dashed violet edge is selected
from $\mathcal{E}^{\mathrm{elg}}(\mathcal{I})$ for reconnection.
(c) Reconnected network, where the selected edge merges the two
stable islands into a single stable reconnected component while
the unstable island remains isolated.}
\label{fig:islanding_example}
\end{figure*}

\subsection{Linear Network Model and Counterexample Search}
\label{subsec:linear-model}

We consider a connected undirected network
$\mathcal{G}=(\mathcal{V},\mathcal{E})$, where each node
$i\in\mathcal{V}$ has state $x_i\in\mathbb{R}^2$ and dynamics
\begin{equation}
    \dot{x}_i
    =
    -w_i x_i
    -
    \sum_{j\in\mathcal{N}_i}w_{ij}(x_i-x_j),
    \label{eq:simulation-node-dynamics}
\end{equation}
where $w_i\in\mathbb{R}$ is the local coefficient of node $i$
and $w_{ij}=w_{ji}\in\mathbb{R}$ is the coupling coefficient
associated with edge $(i,j)$. Let
$D_w\triangleq\operatorname{diag}(w_1,\ldots,w_N)$ and let
$L(\mathcal{E})$ denote the corresponding weighted Laplacian.
By stacking the node states, the overall network dynamics can
be written as
\begin{equation}
    \dot{x}
    =
    -\left[
        \left(D_w+L(\mathcal{E})\right)\otimes \mathbf I_2
    \right]x,
    \label{eq:simulation-network-dynamics}
\end{equation}
where $\mathbf I_2$ denotes the $2\times2$ identity matrix.
The nominal network is constructed using positive node and edge
coefficients. Faults are introduced by changing selected
coefficients from their nominal positive values to negative
values. A node fault changes the corresponding local coefficient
$w_i$, whereas an edge fault changes the corresponding coupling
coefficient $w_{ij}$. Such faults may cause
$D_w+L(\mathcal{E})$ to lose positive definiteness and thereby
destabilize the origin.

For the numerical studies, we choose
$W_i(x_i)=\frac{1}{2}x_i^T x_i$ and
$W_{ij}(x_i,x_j)=0$, which gives $\eta=0$. We set
$\beta=\gamma=1$, so that
$\Omega=\widehat{\Omega}
=\{x\in\mathbb{R}^{2N}:\|x\|_2\leq\sqrt{2}\}$.

We next describe the counterexample search used in the numerical
implementation of Algorithm~\ref{algo:islanding}. For the linear
systems considered here, stability violations can be detected
directly through eigenvalue analysis. Consider a post-islanding
island $I_r=(\mathcal{V}_r,\mathcal{E}_r)$ and define
$H_r\triangleq D_{w,r}+L(\mathcal{E}_r)$, where $D_{w,r}$
denotes the restriction of $D_w$ to the nodes in
$\mathcal{V}_r$. The post-islanding dynamics are
\begin{equation}
    \dot{x}_r
    =
    -\left(H_r\otimes \mathbf I_2\right)x_r.
    \label{eq:simulation-island-dynamics}
\end{equation}
Since $H_r$ is symmetric, the island is asymptotically stable if
and only if $\lambda_{\min}(H_r)>0$. In the numerical
counterexample search, we use a tolerance $\tau=10^{-9}$ and
declare a stability violation whenever an eigenvalue
$\lambda_k$ of $H_r$ satisfies $\lambda_k\leq\tau$.

For each such eigenvalue, let $u_k$ be a corresponding unit
eigenvector and let $e_1=[1,0]^T$. We embed
$u_k\otimes e_1$ into the full network state by setting the states
of all nodes outside $\mathcal{V}_r$ to zero, scale it to the
boundary of $\Omega$, and add it to
$\Omega_{\mathrm{sample}}$.

Counterexamples for the reconnection stage are constructed
analogously. For a reconnected component
$C_s=(\mathcal{V}_s^c,\mathcal{E}_s^c)$, define
$H_s^c\triangleq D_{w,s}+L(\mathcal{E}_s^c)$, where $D_{w,s}$
is the restriction of $D_w$ to $\mathcal{V}_s^c$. The component
is asymptotically stable if and only if
$\lambda_{\min}(H_s^c)>0$. Eigenpairs satisfying
$\lambda_k\leq\tau$ are treated using the same lifting
construction as above, with the resulting directions scaled to
the boundary of $\widehat{\Omega}$ and added to
$\widehat{\Omega}_{\mathrm{sample}}$.

The eigenvalue tests above are used to determine the actual
asymptotic stability of the resulting islands and reconnected
components. In contrast, the stability-penalty terms $F_2$ and
$F_3$ are constructed from the sufficient Lyapunov conditions
derived in Section~\ref{section:problem-formulation}.
Consequently, $F_2$ or $F_3$ may remain positive even when all
resulting components are asymptotically stable. Stability of the
final solution is therefore verified using the eigenvalue tests
described above rather than by requiring the stability penalties
to vanish.

\subsection{Numerical Results}

The network instances are generated as random geometric graphs.
For each network size, node locations are sampled uniformly from
a square region in $\mathbb{R}^2$, and an edge is added between
two nodes if their Euclidean distance is no greater than $2$.
The sampling is repeated until the resulting graph is connected.
The first two studies use the same realization of a $30$-node
network generated over $[0,7.75]^2$, resulting in $82$ edges.

Before faults are introduced, the node and edge coefficients are
independently sampled as $w_i\sim\operatorname{Unif}[1,5]$ and
$w_{ij}\sim\operatorname{Unif}[1,10]$, where $\operatorname{Unif}$ denotes the uniform distribution.. For a faulty node $i$, we set
$w_i=-\sum_{j\in\mathcal{N}_i}w_{ij}^{\mathrm{nom}}-1$, where
$w_{ij}^{\mathrm{nom}}$ denotes the nominal edge coefficient,
while each faulty edge is assigned $w_{ij}=-20$.

Unless otherwise specified, we set
$\alpha_0=1$, $\alpha_1=0.7$, $\alpha_2=\alpha_3=5$,
$c_i=2$, $c_{ij}=1$, $k_1=0.3$, $k_2=0.1$, and
$\hat{\epsilon}=0.01$. 
Both $\Omega_{\mathrm{sample}}$ and
$\widehat{\Omega}_{\mathrm{sample}}$ are initialized with $20$
samples drawn uniformly from the corresponding domains.
All computations were performed on a MacBook Pro
equipped with an Apple M3 Pro chip and 36 GB of unified memory.

\subsubsection{Illustrative Islanding and Reconnection Example}

We first illustrate the proposed islanding and reconnection
procedure on a $30$-node network with $82$ edges. The original
network contains faulty nodes and edges, as shown in
Fig.~\ref{fig:islanding_original}, and is unstable.

Algorithm~\ref{algo:islanding} partitions the network into two
stable islands and one unstable island, as shown in
Fig.~\ref{fig:islanding_post}. The stable islands, shown in blue
and green, have minimum eigenvalues $3.1083$ and $1.9670$,
respectively, while the unstable island is shown in red. Removed
edges are shown as dashed lines, where gray denotes edges that
remain disconnected and violet denotes the edge selected for
reconnection. As shown in
Fig.~\ref{fig:islanding_reconnected}, reconnecting the selected
edge merges the two stable islands into a single reconnected
component while the unstable island remains isolated. The minimum
eigenvalue of the reconnected component is $2.2520>0$, verifying
its asymptotic stability.

\subsubsection{Sensitivity to Fault Configurations}

We next examine the sensitivity of the proposed method to the
locations of the faults. We use the same $30$-node, $82$-edge
network as in the preceding example and consider ten randomly
generated fault configurations. Each configuration contains two
faulty nodes and four faulty edges. The network topology, nominal node and
edge weights, fault-generation rule, and algorithmic parameters
are fixed, while only the fault locations are varied.

Table~\ref{tab:fault_sensitivity} reports the number of inner
optimization solves required for convergence, the total
computation time, and the resulting value of $F_1$. Here, the
number of inner solves denotes the total number of times the
sampled inner optimization problem is solved throughout the
counterexample-guided procedure until convergence. All ten runs
terminate by convergence before reaching the maximum of $20$
outer iterations, and every stable island and reconnected
component in the final solutions satisfies the eigenvalue-based
stability test described in Section~\ref{subsec:linear-model}.
Hence, the reported values of $F_1$ correspond to final solutions
satisfying the numerical stability criterion.

\begin{table}[t]
\centering
\caption{Sensitivity to fault locations on the fixed $30$-node network.}
\label{tab:fault_sensitivity}
\begin{tabular}{@{}ccrr@{}}
\toprule
Config. & Inner solves & Time (s) & $F_1$ \\
\midrule
1  & 8  & 1111.9 & 27.1 \\
2  & 9  &  862.4 & 32.5 \\
3  & 4  &  527.4 & 19.0 \\
4  & 3  &  359.5 & 23.4 \\
5  & 4  &  536.5 & 22.5 \\
6  & 6  &  954.5 & 20.2 \\
7  & 3  &  379.3 & 31.2 \\
8  & 8  &  897.7 & 27.8 \\
9  & 12 & 1457.9 & 22.7 \\
10 & 3  &  342.2 & 35.6 \\
\bottomrule
\end{tabular}
\end{table}

The total computation time ranges from $342.2$~s to
$1457.9$~s, with an average of $742.9$~s, while the number of
inner solves ranges from $3$ to $12$. Longer computation times
are generally associated with a larger number of inner
optimization solves required before convergence, indicating that
the computational effort varies with the fault placement. The
resulting values of $F_1$ range from $19.0$ to $35.6$, indicating
that the fault locations also affect the resulting islanding and
reconnection cost.

\subsubsection{Comparison with the MILP Benchmark}

\begin{table}[t]
\centering
\caption{Comparison of the proposed method with the MILP benchmark.}
\label{tab:milp-comparison}
\begin{tabular}{@{}rrlcrr@{}}
\toprule
$N$ & $M$ & Method & Inner solves & Time (s) & $F_1$ \\
\midrule
10 & 19  & Proposed & 2 &    2.9 &  8.3 \\
   &     & MILP     & 2 &    0.4 &  8.3 \\
\addlinespace
20 & 42  & Proposed & 7 &   88.1 & 21.0 \\
   &     & MILP     & 7 &    5.4 & 17.0 \\
\addlinespace
30 & 82  & Proposed & 4 &  433.2 & 22.3 \\
   &     & MILP     & 4 &   59.6 & 22.3 \\
\addlinespace
40 & 109 & Proposed & 7 & 1490.9 & 37.0 \\
   &     & MILP     & 5 &  239.1 & 25.3 \\
\bottomrule
\end{tabular}
\end{table}

We next compare the proposed method with a MILP benchmark on
four network instances of increasing size. The benchmark replaces
only the inner optimization solver in
Algorithm~\ref{algo:islanding} and is implemented using
Gurobi~13.0.2 \cite{gurobi}. For fixed sample sets, the MILP exactly represents
the corresponding sampled inner optimization problem and is
solved to the prescribed optimality tolerance. The counterexample
search, sample augmentation, and stopping criterion remain
unchanged.

We consider networks with $N\in\{10,20,30,40\}$ and
$M\in\{19,42,82,109\}$, respectively, generated according to the
random geometric graph construction described above. The
$30$-node network is the same realization used in the preceding
two studies. For each network size, the two methods use the same
topology, nominal and faulted coefficients, fault configuration,
and initial sample sets.
Table~\ref{tab:milp-comparison} reports the number of inner
solves, the end-to-end computation time, and the resulting value
of $F_1$. Since the two methods may generate different
counterexamples, their final augmented sample sets need not
coincide. We therefore report $F_1$, which is independent of the
sample sets, rather than the final sampled objective values.
All eight runs terminate by convergence before reaching the
maximum of $20$ outer iterations, and every stable island and
reconnected component satisfies the eigenvalue-based stability
test described in Section~\ref{subsec:linear-model}.
\section{Conclusion}
\label{section:conclusion}
This paper studied controlled islanding and reconnection in
networked dynamical systems. The main challenge arises from the coupling between the two stages,
since the islanding decision determines the available reconnection
options while the reconnected edges may affect the stability of the
resulting components.
We showed that the coupled islanding and reconnection decisions
can be represented as bases of a graphic matroid and that the
resulting optimization problem admits an equivalent nonincreasing
supermodular formulation. This
structure supports a local-search solution method with
counterexample refinement for enforcing the stability requirements.
Numerical studies on linear networked systems demonstrated the
resulting islanding and reconnection strategies under different
fault configurations and quantified their solution quality and
computational performance relative to a MILP benchmark. In this work, an island or reconnected component is stable if its
states converge to the origin. Future work will investigate other
stability or collective-behavior requirements in nonlinear networked
systems, such as synchronization, by deriving the corresponding
sufficient conditions and incorporating them into the islanding and
reconnection formulation.
\bibliographystyle{IEEEtran}
\bibliography{MyBib}
\appendix
This appendix provides the technical details omitted from the main text. 

\subsection{Proof of Lemma~\ref{lemma:valid-islanding}}
\label{app:proof-lemma6}

\begin{proof}
By construction,
$\mathcal{V}_1(A),\ldots,\mathcal{V}_{m(A)}(A)$ form a partition of
$\mathcal{V}$. Hence, the first two conditions in
Definition~\ref{def:islanding} are satisfied.

For any $r\in\{1,\ldots,m(A)\}$ and
$i,j\in\mathcal{V}_r(A)$, there exists a path between $i$ and $j$
using only edges in $\overline{A}$. Let
$i=v_0,v_1,\ldots,v_q=j$ be the nodes in $\mathcal{V}$ encountered
along this path, in order. For each $k\in\{1,\ldots,q\}$, the portion
of the path between $v_{k-1}$ and $v_k$ is
\[
v_{k-1},\;
n_{v_{k-1}}^{v_{k-1}v_k},\;
n_{v_k}^{v_{k-1}v_k},\;
v_k,
\]
and hence $(v_{k-1},v_k)\in\mathcal{E}$. Moreover, each $v_k$ is
connected to $i$ using edges in $\overline{A}$ and therefore belongs
to $\mathcal{V}_r(A)$. It follows that
$(v_{k-1},v_k)\in\mathcal{E}_r(A)$ for every
$k\in\{1,\ldots,q\}$. Thus,
$i=v_0,v_1,\ldots,v_q=j$ is a path in $I_r(A)$, so $I_r(A)$ is
connected. Therefore, $\mathcal{I}(A)$ is a valid islanding of
$\mathcal{G}$.

By Lemma~\ref{lem:graphic_representation}, $A$ can be viewed as a
spanning tree of $\overline{\mathcal{G}}$ containing $e_0$.
For each $r\in\{1,\ldots,m(A)\}$, let $C_r$ denote the connected
component of the subgraph with vertex set
$\overline{\mathcal{V}}\setminus\{s_0,d_0\}$ and edge set
$\overline{A}$ that contains $\mathcal{V}_r(A)$. By construction,
$C_r\cap\mathcal{V}=\mathcal{V}_r(A)$.

Since $A$ is a spanning tree of $\overline{\mathcal{G}}$, there
exists at least one edge of $A$ with one endpoint in $C_r$ and the
other outside $C_r$. By the definition of $\overline{A}$ and
$C_r$, any such edge must connect $C_r$ to either $s_0$ or $d_0$.
Moreover, since $s_0$ and $d_0$ are adjacent only to nodes in
$\mathcal{V}$, this edge is of the form $(s_0,i)$ or $(d_0,i)$
for some $i\in\mathcal{V}_r(A)$.

It remains to show that there cannot be two such edges. Suppose that
two distinct edges connect $C_r$ to $\{s_0,d_0\}$, and let their
endpoints in $C_r$ be $i$ and $j$. Since $C_r$ is connected, there
exists a path in $C_r$ between $i$ and $j$. If the two edges are
both incident to $s_0$ or both incident to $d_0$, this path together
with the two edges forms a cycle in $A$. If one edge is incident to
$s_0$ and the other to $d_0$, this path, the two edges, and
$e_0=(s_0,d_0)$ form a cycle in $A$. Both cases contradict the
acyclicity of $A$. Therefore, $A$ contains exactly one edge in
$\{(s_0,i),(d_0,i):i\in\mathcal{V}_r(A)\}$.
\end{proof}

\subsection{Proof of Theorem~\ref{thm:supermodular_reformulation}}
\label{app:supermodular_reformulation}
We first introduce a property of submodular functions that will be
used repeatedly in this section.

\begin{lemma}\label{lemma:composition}
Let $\mu:2^{\Lambda}\to\mathbb{R}$ be monotone nondecreasing and
submodular, and let $C_0,C_1\subseteq\Lambda$ be fixed. Define
$\nu:2^{\Lambda}\to\mathbb{R}$ by
\begin{equation}
    \nu(\Delta)
    \triangleq
    \mu\big((\Delta\cap C_0)\cup C_1\big),
    \qquad \Delta\subseteq\Lambda.
\end{equation}
Then $\nu$ is monotone nondecreasing and submodular.
\end{lemma}

\begin{proof}
Let $\Delta_1\subseteq\Delta_2\subseteq\Lambda$. Since
$(\Delta_1\cap C_0)\cup C_1
\subseteq
(\Delta_2\cap C_0)\cup C_1$,
the monotonicity of $\mu$ implies
$\nu(\Delta_1)\leq\nu(\Delta_2)$.

To prove submodularity, let
$\xi\in\Lambda\setminus\Delta_2$. If
$\xi\notin C_0\setminus C_1$, then adding $\xi$ does not change either
$(\Delta_1\cap C_0)\cup C_1$ or
$(\Delta_2\cap C_0)\cup C_1$, and hence
\[
\nu(\Delta_1\cup\{\xi\})-\nu(\Delta_1)
=
\nu(\Delta_2\cup\{\xi\})-\nu(\Delta_2)
=
0.
\]
If $\xi\in C_0\setminus C_1$, then
$((\Delta_1\cup\{\xi\})\cap C_0)\cup C_1
=
((\Delta_1\cap C_0)\cup C_1)\cup\{\xi\}$,
and similarly for $\Delta_2$. Since
$(\Delta_1\cap C_0)\cup C_1
\subseteq
(\Delta_2\cap C_0)\cup C_1$
and
$\xi\notin(\Delta_2\cap C_0)\cup C_1$,
the submodularity of $\mu$ gives
\[
\nu(\Delta_1\cup\{\xi\})-\nu(\Delta_1)
\geq
\nu(\Delta_2\cup\{\xi\})-\nu(\Delta_2).
\]
Therefore, $\nu$ is monotone nondecreasing and submodular.
\end{proof}

To obtain a supermodular reformulation of \eqref{eq:matroid-formulation}, we construct extensions of the indicators $\chi_i^s$, $\chi_i^d$, and $\phi_{ij}$ that agree with the original functions on $\mathcal{B}(\mathcal{M}_c)$ and are nonincreasing and supermodular on $2^{\overline{\mathcal{E}}}$.
Let
$Q_i^s \triangleq \{(s_0,d_0)\} \cup \{(s_0,j) : j \in \mathcal{V}\}
\setminus \{(s_0,i)\}$ and
$Q_i^d \triangleq \{(s_0,d_0)\} \cup \{(d_0,j) : j \in \mathcal{V}\}
\setminus \{(d_0,i)\}$. We first define 
\begin{align*}
    \overline{\chi}_i^s(A) & \triangleq \{|\overline{\mathcal{V}}|
        - |A\cap Q_i^s|
        - \rho_{\overline{\mathcal{M}}}\big((A\setminus Q_i^s)\cup\{(s_0,i)\}\big)\}_{+},\\
    \overline{\chi}_i^d(A)& \triangleq \{|\overline{\mathcal{V}}|
        - |A\cap Q_i^d|
        - \rho_{\overline{\mathcal{M}}}\big((A\setminus Q_i^d)\cup\{(d_0,i)\}\big)\}_{+},\\
    \overline{\phi}_{ij}(A)& \triangleq 1 - |A\cap \{(n_i^{ij}, n_j^{ij})\}|,
\end{align*}
with $|\overline{\mathcal{V}}| = N + 2M + 2$.
\begin{lemma}
\label{lemma:chi-relaxed}
    For any $A\in \mathcal{B}(\mathcal{M}_c)$, we have $\overline{\chi}_i^s(A) = {\chi}_i^s(A)$ and $\overline{\chi}_i^d(A) = {\chi}_i^d(A)$ for any $i\in \mathcal{V}$, and $\phi_{ij}(A) = \overline{\phi}_{ij}(A)$ for any $(i,j)\in \mathcal{E}$. Moreover, we have $\overline{\chi}_i^s(A)$, $\overline{\chi}_i^d(A)$, and $\overline{\phi}_{ij}(A)$ are nonincreasing and supermodular on $2^{\overline{\mathcal{E}}}$.
\end{lemma}
\begin{proof}
Let $A\in\mathcal{B}(\mathcal{M}_c)$. We first prove
$\overline{\chi}_i^s(A)=\chi_i^s(A)$. By
Lemma~\ref{lem:graphic_representation}, $A$ is a spanning tree of
$\overline{\mathcal{G}}$, and hence
$|A|=|\overline{\mathcal{V}}|-1$. Since
$A\setminus Q_i^s\subseteq A$, the set $A\setminus Q_i^s$ is
independent in the graphic matroid $\overline{\mathcal{M}}$.
Therefore,
\[
\rho_{\overline{\mathcal{M}}}(A\setminus Q_i^s)
=
|A\setminus Q_i^s|
=
|\overline{\mathcal{V}}|-1-|A\cap Q_i^s|.
\]
Substituting this identity into the definition of
$\overline{\chi}_i^s$ gives
\[
\overline{\chi}_i^s(A)
=
\Big\{
1-
\big[
\rho_{\overline{\mathcal{M}}}
((A\setminus Q_i^s)\cup\{(s_0,i)\})
-
\rho_{\overline{\mathcal{M}}}(A\setminus Q_i^s)
\big]
\Big\}_{+}.
\]
Let
$\delta_i^s(A)
\triangleq
\rho_{\overline{\mathcal{M}}}
((A\setminus Q_i^s)\cup\{(s_0,i)\})
-
\rho_{\overline{\mathcal{M}}}(A\setminus Q_i^s)$.
By Lemma~\ref{lemma:unit-increase},
$\delta_i^s(A)\in\{0,1\}$. Since $(s_0,i)\notin Q_i^s$, if
$(s_0,i)\in A$, then $(s_0,i)\in A\setminus Q_i^s$, and hence
$\delta_i^s(A)=0$.
Now suppose $(s_0,i)\notin A$. By the definition of $Q_i^s$,
every edge incident to $s_0$ other than $(s_0,i)$ belongs to
$Q_i^s$. Hence $s_0$ is an isolated vertex in the subgraph induced
by $A\setminus Q_i^s$. Since $s_0\neq i$, the vertices $s_0$ and
$i$ belong to distinct connected components of this subgraph.
Therefore, adding $(s_0,i)$ increases the rank of the graphic
matroid by one, so $\delta_i^s(A)=1$. By the definition of
$\chi_i^s$, it follows that
$\delta_i^s(A)=1-\chi_i^s(A)$.
Since $\chi_i^s(A)\in\{0,1\}$, we obtain
$\overline{\chi}_i^s(A)=\chi_i^s(A)$.
The same argument, with $s_0$ and $Q_i^s$ replaced by $d_0$ and
$Q_i^d$, respectively, proves
$\overline{\chi}_i^d(A)=\chi_i^d(A)$.

For $(i,j)\in\mathcal{E}$,
$\overline{\phi}_{ij}(A)
=
1-\left|A\cap\{(n_i^{ij},n_j^{ij})\}\right|$.
Thus $\overline{\phi}_{ij}(A)=1$ if and only if
$(n_i^{ij},n_j^{ij})\notin A$, and
$\overline{\phi}_{ij}(A)=0$ otherwise. By the definition of
$\phi_{ij}$, this gives
$\overline{\phi}_{ij}(A)=\phi_{ij}(A)$.

We next establish nonincreasing supermodularity on
$2^{\overline{\mathcal{E}}}$. By
Lemma~\ref{lemma:composition}, the function
$A\mapsto
\rho_{\overline{\mathcal{M}}}
((A\setminus Q_i^s)\cup\{(s_0,i)\})$
is monotone nondecreasing and submodular. Moreover,
$A\mapsto |A\cap Q_i^s|$ is modular and monotone nondecreasing.
Hence
\[
A\mapsto
|\overline{\mathcal{V}}|
-
|A\cap Q_i^s|
-
\rho_{\overline{\mathcal{M}}}
((A\setminus Q_i^s)\cup\{(s_0,i)\})
\]
is nonincreasing and supermodular. By the positive-part truncation
property in Section~\ref{sec:prelim-submodular}, the positive part
of a nonincreasing supermodular function is also nonincreasing and
supermodular. Therefore, $\overline{\chi}_i^s$ is nonincreasing and
supermodular. The same argument applies to
$\overline{\chi}_i^d$.
Finally,
$\overline{\phi}_{ij}(A)
=
1-|A\cap\{(n_i^{ij},n_j^{ij})\}|$
is modular and nonincreasing, and is therefore supermodular.
\end{proof}

For $i \in \mathcal{V}$ and $j \in \overline{\mathcal{V}}\setminus\{s_0,d_0\}$
with $j \neq i$, let $e_{ij}$ be a new edge element with endpoints
$\{i,j\}$, define the labeled multigraph
$\overline{\mathcal{G}}_{ij} \triangleq (\overline{\mathcal{V}},
\overline{\mathcal{E}} \cup \{e_{ij}\}, \overline{\Theta}_{ij})$, and let
$\mathcal{M}_{ij} \triangleq \mathcal{M}(\overline{\mathcal{G}}_{ij})$
denote its graphic matroid. We define
\begin{align*}
\overline{\chi}_{ij}(A) & \triangleq \max\Big\{
    \overline{\chi}_i^s(A) + \overline{\chi}_i^d(A)
    + |\overline{\mathcal{V}}| - |A \cap P|\\
& \qquad\qquad
    - \rho_{\mathcal{M}_{ij}}\big((A \setminus P) \cup \{e_{ij}\}\big),
    \ 1 \Big\} - 1,\\
\overline{\chi}_{ii}(A) & \triangleq
    \overline{\chi}_i^s(A) + \overline{\chi}_i^d(A).
\end{align*}

\begin{lemma}\label{lemma:chi-ij-relaxed}
    For any $A \in \mathcal{B}(\mathcal{M}_{c})$, any $i \in \mathcal{V}$, and any
    $j \in \overline{\mathcal{V}} \setminus \{s_{0}, d_{0}\}$, we have
    $\overline{\chi}_{ij}(A) = \chi_{ij}(A)$. Moreover,
    $\overline{\chi}_{ij}$ is nonincreasing and supermodular on
    $2^{\overline{\mathcal{E}}}$.
\end{lemma}
\begin{proof}
We first consider $j\neq i$. Let
$A\in\mathcal{B}(\mathcal{M}_c)$ and define
\[
\delta_{ij}(A)
\triangleq
\rho_{\mathcal{M}_{ij}}
\big((A\setminus P)\cup\{e_{ij}\}\big)
-
\rho_{\mathcal{M}_{ij}}(A\setminus P).
\]
By Lemma~\ref{lemma:unit-increase},
$\delta_{ij}(A)\in\{0,1\}$. Since $\mathcal{M}_{ij}$ is the
graphic matroid of $\overline{\mathcal{G}}_{ij}$ and $e_{ij}$ is
a new edge with endpoints $i$ and $j$,
$\delta_{ij}(A)=0$ if and only if $i$ and $j$ are connected by a
path in $A\setminus P$.
By Lemma~\ref{lemma:valid-islanding},
$\chi_i^s(A)+\chi_i^d(A)\in\{0,1\}$, and this quantity equals one
if and only if $i$ is a reference node. Hence, by the definition
of $\chi_{ij}$,
\begin{equation}
\chi_{ij}(A)
=
\max\Big\{
\chi_i^s(A)+\chi_i^d(A)+1-\delta_{ij}(A),\,1
\Big\}-1.
\label{eq:chi-ij-rank}
\end{equation}
By Lemma~\ref{lem:graphic_representation}, $A$ is a spanning tree
of $\overline{\mathcal{G}}$. Thus
$|A|=|\overline{\mathcal{V}}|-1$. Moreover,
$A\setminus P$ is acyclic in $\overline{\mathcal{G}}$, and the
edge elements in $A\setminus P$ have the same endpoints in
$\overline{\mathcal{G}}_{ij}$ as in $\overline{\mathcal{G}}$.
Therefore, $A\setminus P$ is independent in $\mathcal{M}_{ij}$,
and
\[
\rho_{\mathcal{M}_{ij}}(A\setminus P)
=
|\overline{\mathcal{V}}|-1-|A\cap P|.
\]
Substituting this identity into \eqref{eq:chi-ij-rank} and using
Lemma~\ref{lemma:chi-relaxed} yields exactly
$\overline{\chi}_{ij}(A)$, proving
$\overline{\chi}_{ij}(A)=\chi_{ij}(A)$.

We next establish nonincreasing supermodularity on
$2^{\overline{\mathcal{E}}}$. Define
\[
g_{ij}(A)
\triangleq
\overline{\chi}_i^s(A)+\overline{\chi}_i^d(A)
+|\overline{\mathcal{V}}|-|A\cap P|
-\rho_{\mathcal{M}_{ij}}
\big((A\setminus P)\cup\{e_{ij}\}\big).
\]
By Lemma~\ref{lemma:chi-relaxed},
$\overline{\chi}_i^s$ and $\overline{\chi}_i^d$ are
nonincreasing and supermodular, while
$A\mapsto-|A\cap P|$ is modular and nonincreasing.
Applying Lemma~\ref{lemma:composition} on
$\overline{\mathcal E}\cup\{e_{ij}\}$ with
$C_0=\overline{\mathcal E}\setminus P$ and
$C_1=\{e_{ij}\}$,
and then restricting to $A\subseteq\overline{\mathcal E}$,
shows that
$A\mapsto
\rho_{\mathcal{M}_{ij}}
\big((A\setminus P)\cup\{e_{ij}\}\big)$
is nondecreasing and submodular. 
Hence $g_{ij}$ is nonincreasing and supermodular. By Section~\ref{sec:prelim-submodular}, taking the maximum of a nonincreasing supermodular function and a constant preserves nonincreasing supermodularity. Therefore, $\max\{g_{ij}(A),1\}$ is nonincreasing and supermodular. Since subtracting a constant preserves both properties, $\overline{\chi}_{ij}$ is nonincreasing and supermodular.
\end{proof}
The extensions constructed above directly handle terms with nonnegative
coefficients. However, $F_1$, $F_2$, and $F_3$ also contain coefficients
that may be negative, for which direct multiplication does not preserve
nonincreasing supermodularity. To handle both cases uniformly, for each
set function
$h:\mathcal{B}(\mathcal{M}_c)\to\{0,1\}$
considered below, we construct nonincreasing supermodular functions
$\overline{h}$ and $\widetilde{h}$ satisfying
\[
\overline{h}(A)=h(A),
\qquad
\widetilde{h}(A)=1-h(A),
\qquad
A\in\mathcal{B}(\mathcal{M}_c).
\]
For any $a\in\mathbb{R}$, define
\[
\mathcal{T}(a;h)(A)
\triangleq
\begin{cases}
a\,\overline{h}(A), & a\geq 0,\\[1mm]
a\bigl(1-\widetilde{h}(A)\bigr), & a<0.
\end{cases}
\]
On $\mathcal{B}(\mathcal{M}_c)$,
$\mathcal{T}(a;h)(A)=ah(A)$. If $a\geq0$, then
$\mathcal{T}(a;h)=a\overline{h}$; if $a<0$, then
$\mathcal{T}(a;h)=a+(-a)\widetilde{h}$. Hence
$\mathcal{T}(a;h)$ is nonincreasing and supermodular on
$2^{\overline{\mathcal{E}}}$ in either case.

We next construct $\widetilde{h}$ for
$\chi_i^s$, $\chi_i^d$, $\chi_{ij}$, and $\phi_{ij}$.
For each $i\in\mathcal V$, define
$\chi_i^0(A)
\triangleq
1-
\left|
A\cap\{(s_0,i),(d_0,i)\}
\right|$,
$\widetilde{\chi}_i^s(A)
\triangleq
\overline{\chi}_i^d(A)+\chi_i^0(A)$, and
$\widetilde{\chi}_i^d(A)
\triangleq
\overline{\chi}_i^s(A)+\chi_i^0(A$).
For $i\in\mathcal V$ and
$j\in\overline{\mathcal V}\setminus\{s_0,d_0\}$, define
$\widetilde{\chi}_{ij}(A)
\triangleq
\sum_{\substack{v=1\\v\neq i}}^N
\overline{\chi}_{vj}(A)$,
and, for each $(i,j)\in\mathcal E$, define
$\widetilde{\phi}_{ij}(A)
\triangleq
\rho_{\overline{\mathcal M}}(\overline{\mathcal E})
-
\rho_{\overline{\mathcal M}}
\left(
A\setminus\{(n_i^{ij},n_j^{ij})\}
\right)$.

\begin{lemma}
\label{lem:complementary-primitive}
For every $A\in\mathcal B(\mathcal M_c)$,
\[
\begin{aligned}
\widetilde{\chi}_i^s(A)
&=1-\chi_i^s(A),&
\widetilde{\chi}_i^d(A)
&=1-\chi_i^d(A),\\
\widetilde{\chi}_{ij}(A)
&=1-\chi_{ij}(A),&
\widetilde{\phi}_{ij}(A)
&=1-\phi_{ij}(A).
\end{aligned}
\]
Moreover,
$\widetilde{\chi}_i^s$,
$\widetilde{\chi}_i^d$,
$\widetilde{\chi}_{ij}$, and
$\widetilde{\phi}_{ij}$
are nonincreasing and supermodular on
$2^{\overline{\mathcal E}}$.
\end{lemma}

\begin{proof}
We first establish the identities on $\mathcal B(\mathcal M_c)$.
By Lemma~\ref{lemma:valid-islanding}, node $i$ is adjacent in $A$
to at most one of $s_0$ and $d_0$. Hence, by definition,
$\chi_i^s(A)+\chi_i^d(A)+\chi_i^0(A)=1$.
Together with Lemma~\ref{lemma:chi-relaxed}, this yields
$\widetilde{\chi}_i^s(A)=1-\chi_i^s(A)$,
and
$\widetilde{\chi}_i^d(A)=1-\chi_i^d(A)$.
Moreover, by Lemma~\ref{lemma:valid-islanding},
each $j\in\overline{\mathcal V}\setminus\{s_0,d_0\}$ is associated
with exactly one reference node, and hence
$\sum_{v=1}^N\chi_{vj}(A)=1$. Therefore, using
Lemma~\ref{lemma:chi-ij-relaxed},
$$\widetilde{\chi}_{ij}(A)
=
\sum_{\substack{v=1\\v\neq i}}^N
\chi_{vj}(A)
=
1-\chi_{ij}(A).$$
Finally, let
$e_{ij}^{m}\triangleq(n_i^{ij},n_j^{ij})$.
By Lemma~\ref{lem:graphic_representation}, $A$ is a basis of
$\overline{\mathcal M}$. Thus
$A\setminus\{e_{ij}^{m}\}$ is independent and
\[
\rho_{\overline{\mathcal M}}
\left(A\setminus\{e_{ij}^{m}\}\right)
=
\rho_{\overline{\mathcal M}}(\overline{\mathcal E})
-
\left|A\cap\{e_{ij}^{m}\}\right|.
\]
It follows that
$\widetilde{\phi}_{ij}(A)
=
\left|A\cap\{e_{ij}^{m}\}\right|
=
1-\phi_{ij}(A)$.

We next establish nonincreasing supermodularity.
Since
\[
\chi_i^0(A)
=
1-\left|A\cap\{(s_0,i),(d_0,i)\}\right|
\]
is modular and nonincreasing, Lemma~\ref{lemma:chi-relaxed}
implies that $\widetilde{\chi}_i^s$ and
$\widetilde{\chi}_i^d$ are nonincreasing and supermodular.
Similarly, Lemma~\ref{lemma:chi-ij-relaxed} implies that
$\widetilde{\chi}_{ij}$ is nonincreasing and supermodular.
Finally, by Lemma~\ref{lemma:composition},
$A\mapsto
\rho_{\overline{\mathcal M}}
\left(A\setminus\{e_{ij}^{m}\}\right)$
is nondecreasing and submodular. Therefore,
$\widetilde{\phi}_{ij}$ is nonincreasing and supermodular.
\end{proof}

The objective functions $F_1$, $F_2$, and $F_3$ also involve
higher-order set functions recursively constructed through
$\mathcal C$. We next show that the existence of
$\overline{h}$ and $\widetilde{h}$ is preserved under
$\Gamma$, allowing the same construction to be applied
recursively.

\begin{lemma}
\label{lem:C-extension}
Let $q\in\mathbb N_{+}$ and
$h_1,\ldots,h_q:
\mathcal B(\mathcal M_c)\to\{0,1\}$.
Suppose that, for each $r\in\{1,\ldots,q\}$, there exist
nonincreasing supermodular functions
$\overline{h}_r,\widetilde{h}_r:
2^{\overline{\mathcal E}}\to\mathbb R$
such that, for every $A\in\mathcal B(\mathcal M_c)$,
$\overline{h}_r(A)=h_r(A)$,
and
$\widetilde{h}_r(A)=1-h_r(A)$.
Define
\begin{align*}
h(A)
&\triangleq
\Gamma(h_1,\ldots,h_q)(A),\\
\overline{h}(A)
&\triangleq
\Gamma
\left(
\overline{h}_1,\ldots,\overline{h}_q
\right)(A),\\
\widetilde{h}(A)
&\triangleq
\sum_{r=1}^{q}
\Gamma
\left(
\overline{h}_1,\ldots,
\overline{h}_{r-1},
\widetilde{h}_r
\right)(A).
\end{align*}
Then $\overline{h}$ and $\widetilde{h}$ are nonincreasing and
supermodular on $2^{\overline{\mathcal E}}$, and, for every
$A\in\mathcal B(\mathcal M_c)$,
$\overline{h}(A)=h(A)$,
and 
$\widetilde{h}(A)=1-h(A)$.

\end{lemma}

\begin{proof}
Let $A\in\mathcal B(\mathcal M_c)$. Since each $h_r(A)$ is binary,
the definition of $\Gamma$ gives
$\Gamma(h_1,\ldots,h_q)(A)
=
\prod_{r=1}^{q}h_r(A)$.
Using $\overline{h}_r(A)=h_r(A)$ therefore yields
$\overline{h}(A)=h(A)$.
Similarly, for each $r$,
$\Gamma
\left(
\overline{h}_1,\ldots,
\overline{h}_{r-1},
\widetilde{h}_r
\right)(A)
=
\left(\prod_{s=1}^{r-1}h_s(A)\right)
\left(1-h_r(A)\right)$,
where the empty product for $r=1$ equals one. Hence
\[
\begin{aligned}
\widetilde{h}(A)
&=
\sum_{r=1}^{q}
\left(\prod_{s=1}^{r-1}h_s(A)\right)
\left(1-h_r(A)\right)\\
&=
1-\prod_{r=1}^{q}h_r(A)
=
1-h(A),
\end{aligned}
\]
where the second equality follows by telescoping.

It remains to establish nonincreasing supermodularity.
Since each $\overline{h}_r$ is nonincreasing and supermodular,
\[
A\mapsto
\sum_{r=1}^{q}\overline{h}_r(A)-(q-1)
\]
is also nonincreasing and supermodular. By the definition of
$\Gamma$ and Section~\ref{sec:prelim-submodular},
$\overline{h}$ is therefore nonincreasing and supermodular.
Similarly, for each $r$,
$\overline{h}_1,\ldots,\overline{h}_{r-1},
\widetilde{h}_r$ are nonincreasing and supermodular. Thus, by the
same argument,
$\Gamma
\left(
\overline{h}_1,\ldots,
\overline{h}_{r-1},
\widetilde{h}_r
\right)$
is nonincreasing and supermodular. Since $\widetilde{h}$ is the
sum of these functions, it is also nonincreasing and supermodular.
\end{proof}

We now apply the preceding constructions to $F_1$. Instead of
constructing a separate extension of $\chi_{ij}^{\mathrm{rec}}$, we
directly define
\begin{align}
&\overline{F}_1(A)
\triangleq 
\sum_{i=1}^{N}
\sum_{v=1}^{N}
c_i\,
\Gamma
\left(
    \overline{\chi}_v^d,
    \overline{\chi}_{vi}
\right)(A)
\nonumber\\
&+
\frac{\alpha_0-\alpha_1}{2}
\sum_{i=1}^{N}
\sum_{j\in\mathcal N_i}
\sum_{v=1}^{N}
\sum_{\substack{v^{{\prime}}=1\\v^{{\prime}}\neq v}}^{N}
c_{ij}\,
\Gamma
\left(
    \overline{\chi}_{vi},
    \overline{\chi}_{v^{{\prime}}j}
\right)(A)
\nonumber\\
&+
\frac{\alpha_1}{2}
\sum_{i=1}^{N}
\sum_{j\in\mathcal N_i}
\sum_{v=1}^{N}
\sum_{\substack{v^{{\prime}}=1\\v^{{\prime}}\neq v}}^{N}
c_{ij}
\Big[
\Gamma
\left(
    \overline{\chi}_{vi},
    \overline{\chi}_{v^{{\prime}}j},
    \widetilde{\chi}_v^s
\right)(A)
\nonumber\\
&\qquad+
\Gamma
\left(
    \overline{\chi}_{vi},
    \overline{\chi}_{v^{{\prime}}j},
    \overline{\chi}_v^s,
    \widetilde{\chi}_{v^{{\prime}}}^s
\right)(A)
\nonumber\\
&\qquad+
\Gamma
\left(
    \overline{\chi}_{vi},
    \overline{\chi}_{v^{{\prime}}j},
    \overline{\chi}_v^s,
    \overline{\chi}_{v^{{\prime}}}^s,
    \widetilde{\phi}_{ij}
\right)(A)
\Big].
\label{eq:F1bar}
\end{align}
We are now ready to prove 
Theorem~\ref{thm:supermodular_reformulation}.

\begin{proof}[Proof of Theorem~\ref{thm:supermodular_reformulation}]
Fix $A\in\mathcal B(\mathcal M_c)$. By the definitions of
$\psi_{v,i}^s$ and $\chi_{ij}^{\mathrm{rec}}$,
\[
\chi_{ij}^{\mathrm{rec}}(A)
=
\sum_{v=1}^{N}
\sum_{\substack{v^{{\prime}}=1\\v^{{\prime}}\neq v}}^{N}
\chi_{vi}(A)
\chi_{v^{{\prime}}j}(A)
\chi_v^s(A)
\chi_{v^{{\prime}}}^s(A)
\phi_{ij}(A).
\]
Substituting this identity into $F_1$ gives
\begin{align}
F_1(A)
={}&
\sum_{i=1}^{N}
\sum_{v=1}^{N}
c_i\,
\Gamma
\left(
    \chi_v^d,
    \chi_{vi}
\right)(A)
\nonumber\\
&+
\frac{1}{2}
\sum_{i=1}^{N}
\sum_{j\in\mathcal N_i}
\sum_{v=1}^{N}
\sum_{\substack{v^{{\prime}}=1\\v^{{\prime}}\neq v}}^{N}
c_{ij}
\chi_{vi}(A)
\chi_{v^{{\prime}}j}(A)
\nonumber\\
&\qquad\cdot
\left(
    \alpha_0
    -
    \alpha_1
    \chi_v^s(A)
    \chi_{v^{{\prime}}}^s(A)
    \phi_{ij}(A)
\right).
\label{eq:F1-expanded}
\end{align}
For each summand in \eqref{eq:F1-expanded},
\begin{align*}
&\alpha_0
-\alpha_1
\chi_v^s(A)\chi_{v^{{\prime}}}^s(A)\phi_{ij}(A)
\\
={}&
(\alpha_0-\alpha_1)
+\alpha_1
\Big[
(1-\chi_v^s(A))
+\chi_v^s(A)(1-\chi_{v^{{\prime}}}^s(A))
\\
&\hspace{3.2cm}
+\chi_v^s(A)\chi_{v^{{\prime}}}^s(A)
(1-\phi_{ij}(A))
\Big].
\end{align*}
By Lemmas~\ref{lemma:chi-relaxed},
\ref{lemma:chi-ij-relaxed}, and
\ref{lem:complementary-primitive}, the functions
$\overline{\chi}$ and $\widetilde{\chi}$ appearing in
\eqref{eq:F1bar} agree on $\mathcal B(\mathcal M_c)$ with the
corresponding indicators and their complements, respectively.
Since these quantities are binary on $\mathcal B(\mathcal M_c)$,
the definition of $\Gamma$ implies that every
$\Gamma$ term in \eqref{eq:F1bar} agrees with the corresponding
product above. Comparing \eqref{eq:F1bar} and
\eqref{eq:F1-expanded} therefore yields
$\overline F_1(A)=F_1(A)$.

We next establish that $\overline F_1$ is nonincreasing and
supermodular. By Lemmas~\ref{lemma:chi-relaxed},
\ref{lemma:chi-ij-relaxed}, and
\ref{lem:complementary-primitive}, all functions appearing as
arguments of $\Gamma$ in \eqref{eq:F1bar} are nonincreasing
and supermodular on $2^{\overline{\mathcal E}}$. By the definition
of $\Gamma$ and Section~\ref{sec:prelim-submodular}, every
$\Gamma$ term in \eqref{eq:F1bar} has the same properties.
Since $c_i,c_{ij}\geq0$ and $0<\alpha_1<\alpha_0$, all coefficients
in \eqref{eq:F1bar} are nonnegative. Hence $\overline F_1$ is
nonincreasing and supermodular on
$2^{\overline{\mathcal E}}$.

We next consider $F_2$. Unlike $F_1$, the coefficients multiplying
the indicator functions in $\Phi_v(x,A)$ depend on $x$ and may have
either sign. We therefore use the signed-coefficient construction
$\mathcal T$ introduced above.
Recall that
\begin{align*}
    \psi_{v,i}^s
&=
\Gamma
\left(
    \chi_v^s,
    \chi_{vi}
\right),\\
\psi_{v,i,j}^s
&=
\Gamma
\left(
    \chi_v^s,
    \chi_{vi},
    \chi_{vj}
\right),\\
\psi_{v,i,j,k}^s
&=
\Gamma
\left(
    \chi_v^s,
    \chi_{vi},
    \chi_{vj},
    \chi_{vk}
\right).
\end{align*}
By Lemma~\ref{lem:C-extension}, there exist nonincreasing supermodular functions
$\overline{\psi}_{v,i}^s,\;
\widetilde{\psi}_{v,i}^s,\;
\overline{\psi}_{v,i,j}^s,\;
\widetilde{\psi}_{v,i,j}^s,\;
\overline{\psi}_{v,i,j,k}^s,\;
\widetilde{\psi}_{v,i,j,k}^s
:
2^{\overline{\mathcal E}}\to\mathbb R$
such that, for every
$A\in\mathcal B({\mathcal M}_c)$,
\begin{align*}
\overline{\psi}_{v,i}^s(A)
&=
\psi_{v,i}^s(A),
&
\widetilde{\psi}_{v,i}^s(A)
&=
1-\psi_{v,i}^s(A),
\\
\overline{\psi}_{v,i,j}^s(A)
&=
\psi_{v,i,j}^s(A),
&
\widetilde{\psi}_{v,i,j}^s(A)
&=
1-\psi_{v,i,j}^s(A),
\\
\overline{\psi}_{v,i,j,k}^s(A)
&=
\psi_{v,i,j,k}^s(A),
&
\widetilde{\psi}_{v,i,j,k}^s(A)
&=
1-\psi_{v,i,j,k}^s(A).
\end{align*}

We first define a function corresponding to $W_v(x,A)$ as
\begin{align}
\overline W_v(x,A)
\triangleq {}&
\sum_{i=1}^{N}
\overline{\psi}_{v,i}^s(A)W_i(x_i)
\nonumber\\
&+
\frac{1}{2}
\sum_{i=1}^{N}
\sum_{j\in\mathcal N_i}
\overline{\psi}_{v,i,j}^s(A)
W_{ij}(x_i,x_j).
\label{eq:Wr-bar}
\end{align}
Since $W_i(x_i)\geq0$ and $W_{ij}(x_i,x_j)\geq0$,
$\overline W_v(x,A)$ is nonincreasing and supermodular in $A$
for every fixed $x$. Moreover, for every
$A\in\mathcal B({\mathcal M}_c)$,
$\overline W_v(x,A)=W_v(x,A)$.

For notational convenience, define
\begin{align*}
\zeta_i(x)
&\triangleq
\frac{\partial W_i(x_i)}{\partial x_i}f_i(x_i),\\
\zeta_{ij}(x)
&\triangleq
\frac{\partial W_{ij}(x_i,x_j)}{\partial x_i}f_i(x_i)
+
\frac{\partial W_i(x_i)}{\partial x_i}
f_{ij}(x_i,x_j),\\
\zeta_{ijk}(x)
&\triangleq
\frac{\partial W_{ij}(x_i,x_j)}{\partial x_i}
f_{ik}(x_i,x_k).
\end{align*}

Using the signed-coefficient construction $\mathcal T$, define
\begin{align}
&\overline{\Phi}_v(x,A)
\triangleq {}
\sum_{i=1}^{N}
\mathcal T
\left(
    \zeta_i(x);
    \psi_{v,i}^s
\right)(A)
\nonumber\\
&+
\sum_{i=1}^{N}
\sum_{j\in\mathcal N_i}
\mathcal T
\left(
    \zeta_{ij}(x);
    \psi_{v,i,j}^s
\right)(A)
\nonumber\\
&+
\sum_{i=1}^{N}
\sum_{j\in\mathcal N_i}
\sum_{k\in\mathcal N_i}
\mathcal T
\left(
    \zeta_{ijk}(x);
    \psi_{v,i,j,k}^s
\right)(A)
\nonumber\\
&+
k_1\overline W_v(x,A).
\label{eq:Phi-bar}
\end{align}
We then define
\begin{equation}
\overline F_2(A)
\triangleq
\int_{\Omega}
\sum_{v=1}^{N}
\left\{
    \overline{\Phi}_v(x,A)
\right\}_{+}
\,dx.
\label{eq:F2-bar}
\end{equation}
Fix $A\in\mathcal B({\mathcal M}_c)$ and $x\in\Omega$.
By the construction of $\mathcal T$, for any
$a\in\mathbb R$ and any binary set function $h$ for which the
corresponding extensions have been constructed,
$\mathcal T(a;h)(A)=a h(A)$.
Hence,
\begin{align*}
\mathcal T
\left(
    \zeta_i(x);
    \psi_{v,i}^s
\right)(A)
&=
\zeta_i(x)\psi_{v,i}^s(A),
\\
\mathcal T
\left(
    \zeta_{ij}(x);
    \psi_{v,i,j}^s
\right)(A)
&=
\zeta_{ij}(x)\psi_{v,i,j}^s(A),
\\
\mathcal T
\left(
    \zeta_{ijk}(x);
    \psi_{v,i,j,k}^s
\right)(A)
&=
\zeta_{ijk}(x)\psi_{v,i,j,k}^s(A).
\end{align*}
Moreover, by \eqref{eq:Wr-bar}, we have
$\overline W_v(x,A)=W_v(x,A)$.
Substituting these identities into \eqref{eq:Phi-bar} gives
\[
\overline{\Phi}_v(x,A)=\Phi_v(x,A).
\]
Therefore, for every
$A\in\mathcal B({\mathcal M}_c)$,
\begin{align*}
\overline F_2(A)
&=
\int_{\Omega}
\sum_{v=1}^{N}
\left\{
    \overline{\Phi}_v(x,A)
\right\}_{+}
\,dx
\\
&=
\int_{\Omega}
\sum_{v=1}^{N}
\left\{
    \Phi_v(x,A)
\right\}_{+}
\,dx
=
F_2(A).
\end{align*}

It remains to establish nonincreasing supermodularity.
For each fixed $x\in\Omega$, the construction of $\mathcal T$
implies that
$\mathcal T
\left(
    \zeta_i(x);
    \psi_{v,i}^s
\right)$,
$\mathcal T
\left(
    \zeta_{ij}(x);
    \psi_{v,i,j}^s
\right)$,
and
$\mathcal T
\left(
    \zeta_{ijk}(x);
    \psi_{v,i,j,k}^s
\right)$
are nonincreasing and supermodular on
$2^{\overline{\mathcal E}}$.
Furthermore, $\overline W_v(x,A)$ is nonincreasing and
supermodular in $A$, and $k_1>0$.
It follows from \eqref{eq:Phi-bar} that
$\overline{\Phi}_v(x,A)$ is nonincreasing and supermodular in $A$
for every fixed $x\in\Omega$.

Consequently,
$\left\{
    \overline{\Phi}_v(x,A)
\right\}_{+}
=
\max
\left\{
    \overline{\Phi}_v(x,A),0
\right\}$
is nonincreasing and supermodular in $A$.
Finite sums preserve both properties, and integration over the
fixed domain $\Omega$ preserves the corresponding inequalities.
Therefore, $\overline F_2$ is nonincreasing and supermodular on
$2^{\overline{\mathcal E}}$.

We next consider $F_3$. The functions
$\psi_{v,i,k}^{r}$,
$\psi_{v,i,j,k}^{sr}$, and
$\psi_{v,i,j,k}^{rr}$ appearing in $\Psi_v(x,A)$ involve the
reconnection indicators. Rather than constructing extensions of the
reconnection indicators separately, we first rewrite these functions
as sums of conjunctions of the basic binary set functions.
Fix $A\in\mathcal B({\mathcal M}_c)$. By the definition of
$\chi_{ik}^{\mathrm{rec}}$ and the uniqueness of the reference node
associated with each island, we have
\begin{align*}
\psi_{v,i,k}^{r}(A)
&=
\sum_{\substack{v^{\prime}=1\\
v^{\prime}\neq v}}^{N}
\Gamma
\left(
    \chi_v^s,
    \chi_{vi},
    \chi_{v^{\prime}k},
    \chi_{v^{\prime}}^s,
    \phi_{ik}
\right)(A),
\\
\psi_{v,i,j,k}^{sr}(A)
&=
\sum_{\substack{v^{\prime}=1\\
v^{\prime}\neq v}}^{N}
\Gamma
\left(
    \chi_v^s,
    \chi_{vi},
    \chi_{vj},
    \chi_{v^{\prime}k},
    \chi_{v^{\prime}}^s,
    \phi_{ik}
\right)(A),
\\
\psi_{v,i,j,k}^{rr}(A)
&=
\sum_{\substack{v^{\prime}=1\\
v^{\prime}\neq v}}^{N}
\sum_{\substack{v^{\prime\prime}=1\\
v^{\prime\prime}\neq v}}^{N}
\Gamma
\Big(
    \chi_v^s,
    \chi_{vi},
    \chi_{v^{\prime}j},
    \chi_{v^{\prime}}^s,
    \phi_{ij},
\\
&\quad\quad\quad\quad\quad\quad\quad\quad\quad\quad
    \chi_{v^{\prime\prime}k},
    \chi_{v^{\prime\prime}}^s,
    \phi_{ik}
\Big)(A).
\end{align*}
Note that $v^{\prime}$ and $v^{\prime\prime}$ are not required
to be distinct, since $j$ and $k$ may belong to the same stable
island.
By Lemma~\ref{lem:C-extension}, every $\Gamma$-based conjunction appearing above
admits nonincreasing supermodular extensions satisfying the
requirements of the signed-coefficient construction $\mathcal T$.
To keep the construction of $\overline\Psi_v$ compact, for any
$a\in\mathbb R$ define
\begin{align*}
\mathcal R^r_{v,i,k}(a;A)
\triangleq{}&
\sum_{\substack{v^{\prime}=1\\v^{\prime}\neq v}}^N
\mathcal T\!\left(
    a;
    \Gamma\!\left(
        \chi_v^s,
        \chi_{vi},
        \chi_{v^{\prime}k},
        \chi_{v^{\prime}}^s,
        \phi_{ik}
    \right)
\right)(A),
\\
\mathcal R^{sr}_{v,i,j,k}(a;A)
\triangleq{}&
\sum_{\substack{v^{\prime}=1\\v^{\prime}\neq v}}^N
\mathcal T\!\left(
    a;
    \Gamma\!\left(
        \chi_v^s,
        \chi_{vi},
        \chi_{vj},
        \chi_{v^{\prime}k},
        \chi_{v^{\prime}}^s,
        \phi_{ik}
    \right)
\right)(A),
\\
\mathcal R^{rr}_{v,i,j,k}(a;A)
\triangleq{}&
\sum_{\substack{v^{\prime}=1\\v^{\prime}\neq v}}^N
\sum_{\substack{v^{\prime\prime}=1\\v^{\prime\prime}\neq v}}^N
\mathcal T\!\Big(
    a;
    \Gamma\!\big(
        \chi_v^s,
        \chi_{vi},
        \chi_{v^{\prime}j},
        \chi_{v^{\prime}}^s,
\nonumber\\
&\quad\quad\quad\quad\quad\quad
        \phi_{ij},
        \chi_{v^{\prime\prime}k},
        \chi_{v^{\prime\prime}}^s,
        \phi_{ik}
    \big)
\Big)(A).
\end{align*}
By Lemma~\ref{lem:C-extension} and the definition of $\mathcal T$,
$\mathcal R^r_{v,i,k}(a;A)$,
$\mathcal R^{sr}_{v,i,j,k}(a;A)$, and
$\mathcal R^{rr}_{v,i,j,k}(a;A)$
are nonincreasing and supermodular in $A$ for every fixed
$a\in\mathbb R$. Moreover, for every
$A\in\mathcal B(\mathcal M_c)$, we have
\begin{align*}
    \mathcal R^r_{v,i,k}(a;A)
&=
a\psi^r_{v,i,k}(A),
\\
\mathcal R^{sr}_{v,i,j,k}(a;A)
&=
a\psi^{sr}_{v,i,j,k}(A),
\\
\mathcal R^{rr}_{v,i,j,k}(a;A)
&=
a\psi^{rr}_{v,i,j,k}(A).
\end{align*}

For notational convenience, define
\begin{align*}
\zeta^r_{ik}(x)
\triangleq{}&
\frac{\partial W_i(x_i)}{\partial x_i}f_{ik}(x_i,x_k)
+
\frac{\partial W_{ik}(x_i,x_k)}{\partial x_i}f_i(x_i),
\\
\zeta^{sr}_{ijk}(x)
\triangleq{}&
\frac{\partial W_{ij}(x_i,x_j)}{\partial x_i}
f_{ik}(x_i,x_k)
\\[-1mm]
&+
\frac{\partial W_{ik}(x_i,x_k)}{\partial x_i}
f_{ij}(x_i,x_j),
\\
\zeta^{rr}_{ijk}(x)
\triangleq{}&
\frac{\partial W_{ij}(x_i,x_j)}{\partial x_i}
f_{ik}(x_i,x_k).
\end{align*}
We then define
\begin{equation}
\overline F_3(A)
\triangleq
\int_{\widehat\Omega}
\sum_{v=1}^N
\left\{\overline\Psi_v(x,A)\right\}_{+}\,dx.
\label{eq:F3-bar}
\end{equation}
where
\begin{align}
&\overline{\Psi}_v(x,A)
\triangleq{}
\sum_{i=1}^{N}
\sum_{k\in\mathcal N_i}
\mathcal R^r_{v,i,k}
\left(
    \zeta_{ik}^{r}(x);A
\right)
\nonumber\\
&+
\sum_{i=1}^{N}
\sum_{j\in\mathcal N_i}
\sum_{k\in\mathcal N_i}
\mathcal R^{sr}_{v,i,j,k}
\left(
    \zeta_{ijk}^{sr}(x);A
\right)
\nonumber\\
&+
\sum_{i=1}^{N}
\sum_{j\in\mathcal N_i}
\sum_{k\in\mathcal N_i}
\mathcal R^{rr}_{v,i,j,k}
\left(
    \zeta_{ijk}^{rr}(x);A
\right)
\nonumber\\
&+
\sum_{i=1}^{N}
\sum_{k\in\mathcal N_i}
\mathcal R^r_{v,i,k}
\left(
    \frac{k_2}{2}W_{ik}(x_i,x_k);A
\right)
\nonumber\\
&+
\sum_{i=1}^{N}
\mathcal T\!\left(
    -(k_1-k_2)W_i(x_i);
    \psi_{v,i}^s
\right)(A)
\nonumber\\
&+
\frac{1}{2}
\sum_{i=1}^{N}
\sum_{j\in\mathcal N_i}
\nonumber\\[-1mm]
&\qquad
\mathcal T\!\left(
    -(k_1-k_2)W_{ij}(x_i,x_j);
    \psi_{v,i,j}^s
\right)(A).
\label{eq:Psi-bar}
\end{align}

Fix $A\in\mathcal B({\mathcal M}_c)$ and
$x\in\widehat{\Omega}$.
By the construction of $\mathcal T$,
$\mathcal T(a;h)(A)=a h(A)$
for every binary set function $h$ considered above and every
$a\in\mathbb R$.
Applying this identity to the first three groups in
\eqref{eq:Psi-bar}, together with the preceding representations of
$\psi_{v,i,k}^{r}$,
$\psi_{v,i,j,k}^{sr}$, and
$\psi_{v,i,j,k}^{rr}$, gives
\begin{align*}
&\sum_{\substack{v^{\prime}=1\\
v^{\prime}\neq v}}^{N}
\mathcal T
\left(
    \zeta_{ik}^{r}(x);
    \Gamma
    \left(
        \chi_v^s,
        \chi_{vi},
        \chi_{v^{\prime}k},
        \chi_{v^{\prime}}^s,
        \phi_{ik}
    \right)
\right)(A)
\\
&\qquad=
\zeta_{ik}^{r}(x)
\psi_{v,i,k}^{r}(A),
\end{align*}
and analogously,
\begin{align*}
\sum_{\substack{v^{\prime}=1\\
v^{\prime}\neq v}}^{N}
\mathcal T
\left(
    \zeta_{ijk}^{sr}(x);
    \Gamma(\cdot)
\right)(A)
&=
\zeta_{ijk}^{sr}(x)
\psi_{v,i,j,k}^{sr}(A),
\\
\sum_{\substack{v^{\prime}=1\\
v^{\prime}\neq v}}^{N}
\sum_{\substack{v^{\prime\prime}=1\\
v^{\prime\prime}\neq v}}^{N}
\mathcal T
\left(
    \zeta_{ijk}^{rr}(x);
    \Gamma(\cdot)
\right)(A)
&=
\zeta_{ijk}^{rr}(x)
\psi_{v,i,j,k}^{rr}(A),
\end{align*}
where $\Gamma(\cdot)$ denotes the corresponding conjunction in the
definitions above.

Moreover, since the extensions agree with the original indicators on
$\mathcal B(\mathcal M_c)$, we have
\begin{align*}
&\sum_{\substack{v^{\prime}=1\\
v^{\prime}\neq v}}^{N}
\Gamma
\left(
    \overline{\chi}_{v}^s,
    \overline{\chi}_{vi},
    \overline{\chi}_{v^{\prime}k},
    \overline{\chi}_{v^{\prime}}^s,
    \overline{\phi}_{ik}
\right)(A)
=
\psi_{v,i,k}^{r}(A).
\end{align*}
Finally, we have
\begin{align*}
&\mathcal{T}
\big(
    -(k_1-k_2)W_i(x_i);
    \psi_{v,i}^s
\big)(A)
\\
&\quad\quad\quad\quad\quad=
-(k_1-k_2)
W_i(x_i)\psi_{v,i}^s(A),
\\
&\mathcal T
\big(
    -(k_1-k_2)W_{ij}(x_i,x_j);
    \psi_{v,i,j}^s
\big)(A)
\\
&\quad\quad\quad\quad\quad=
-(k_1-k_2)
W_{ij}(x_i,x_j)\psi_{v,i,j}^s(A).
\end{align*}
Substituting these identities into \eqref{eq:Psi-bar} yields
\[
\overline{\Psi}_v(x,A)=\Psi_v(x,A).
\]
Consequently,
\[
\overline F_3(A)
=
\int_{\widehat{\Omega}}
\sum_{v=1}^{N}
\left\{
    \Psi_v(x,A)
\right\}_{+}
\,dx
=
F_3(A).
\]

It remains to establish nonincreasing supermodularity.
By Lemma~\ref{lem:C-extension}, every conjunction appearing in
the first three groups of \eqref{eq:Psi-bar} admits the extensions
required by $\mathcal T$. Hence, for every fixed
$x\in\widehat{\Omega}$, each corresponding $\mathcal T$ term is
nonincreasing and supermodular on
$2^{\overline{\mathcal E}}$.

The terms in the fourth group of \eqref{eq:Psi-bar} are
nonincreasing and supermodular in $A$ by the properties of
$\mathcal{R}^{r}_{v,i,k}(a;A)$ established above.
The last two groups are nonincreasing and supermodular by the
definition of $\mathcal T$. Therefore,
$\overline{\Psi}_v(x,\cdot)$ is nonincreasing and supermodular on
$2^{\overline{\mathcal E}}$ for every fixed
$x\in\widehat{\Omega}$.
It follows that
$\left\{
    \overline{\Psi}_v(x,A)
\right\}_{+}
=
\max
\left\{
    \overline{\Psi}_v(x,A),0
\right\}$
is nonincreasing and supermodular in $A$.
Finite sums preserve both properties, and integration over the fixed
domain $\widehat{\Omega}$ preserves the corresponding inequalities.
Hence $\overline F_3$ is nonincreasing and supermodular on
$2^{\overline{\mathcal E}}$.

Finally, since $\alpha_2,\alpha_3\geq0$, the function $\overline F$ defined in Theorem~\ref{thm:supermodular_reformulation} is a nonnegative weighted sum of $\overline F_1$, $\overline F_2$, and $\overline F_3$, and is therefore nonincreasing and supermodular on $2^{\overline{\mathcal E}}$. Moreover, for every $A\in\mathcal B(\mathcal M_c)$, $\overline F(A) = F_1(A)+\alpha_2F_2(A)+\alpha_3F_3(A) = F(A)$. 
\end{proof}

\end{document}